\pdfoutput=1
\documentclass[a4paper,reprint,onecolumn,notitlepage,aip,nofootinbib]{revtex4-2}
\usepackage[left=1.5cm,right=1.5cm,top=1cm,bottom=1.5cm,includeheadfoot]{geometry}

\usepackage[english]{babel}
\usepackage[T1]{fontenc}
\usepackage{tgtermes} 
\usepackage{tgheros} 

\usepackage{amssymb}
\usepackage{amsmath}
\usepackage{amsthm}
\usepackage{mathtools}
\usepackage{braket}
\usepackage{mathrsfs} 
\usepackage{enumerate}
\usepackage[dvipsnames]{xcolor}
\usepackage[shortlabels]{enumitem}
\usepackage[normalem]{ulem}
\PassOptionsToPackage{hyphens}{url} 
\usepackage[breaklinks=true]{hyperref}
\bibpunct{[}{]}{;}{n}{}{} 

\usepackage{tikz}
\usetikzlibrary{patterns}

\theoremstyle{plain}
\newtheorem{theorem}{Theorem}

\newtheorem{lemma}[theorem]{Lemma}
\newtheorem{corollary}[theorem]{Corollary}
\newtheorem{definition}[theorem]{Definition}

\theoremstyle{remark}

\newtheorem{example}{Example}

\newcommand{\R}{\mathbb{R}}
\newcommand{\C}{\mathbb{C}}
\newcommand{\N}{\mathbb{N}}

\newcommand{\T}{\mathbb{T}}

\renewcommand{\d}{\mathrm{d}}

\renewcommand{\i}{\mathrm{i}}

\newcommand{\eps}{\varepsilon}
\renewcommand{\H}{\mathcal{H}}
\renewcommand{\Re}{\mathop\mathsf{Re}}
\renewcommand{\Im}{\mathop\mathsf{Im}}
\newcommand{\FDM}{F_\mathrm{DM}}
\newcommand{\FLieb}{F_\text{Lieb}}
\DeclareMathOperator{\Tr}{Tr}
\newcommand{\vext}{v_\mathrm{ext}}
\newcommand{\vH}{v_\mathrm{H}}
\newcommand{\vxc}{v_\mathrm{xc}}
\newcommand{\rhoxc}{\rho_\mathrm{xc}}

\newcommand{\densset}{\mathscr{I}}
\newcommand{\psiset}{\mathscr{W}}
\newcommand{\affspace}{\mathscr{X}}
\newcommand{\vrepset}{\affspace_{>0}}

\DeclarePairedDelimiter{\abs}{\lvert}{\rvert}
\DeclarePairedDelimiter{\norm}{\Vert}{\rVert}
\DeclarePairedDelimiter{\innerproduct}{\langle}{\rangle}

\begin{document}
\title{Solution of the \texorpdfstring{$v$}{v}-representability problem on a one-dimensional torus}

\author{Sarina M. Sutter}
\address{Theoretical Chemistry, Faculty of Exact Sciences, VU University, Amsterdam, The Netherlands}

\author{Markus Penz}
\email[Electronic address:\;]{m.penz@inter.at}
\affiliation{Max Planck Institute for the Structure and Dynamics of Matter and Center for Free-Electron Laser Science, Hamburg, Germany}
\affiliation{Department of Computer Science, Oslo Metropolitan University, Oslo, Norway}

\author{Michael Ruggenthaler}
\affiliation{Max Planck Institute for the Structure and Dynamics of Matter and Center for Free-Electron Laser Science, Hamburg, Germany}
\affiliation{The Hamburg Center for Ultrafast Imaging, Hamburg, Germany}

\author{Robert van Leeuwen}
\address{Department of Physics, Nanoscience Center, University of Jyv\"askyl\"a, Jyv\"askyl\"a, Finland}

\author{Klaas J. H. Giesbertz}
\address{Theoretical Chemistry, Faculty of Exact Sciences, VU University, Amsterdam, The Netherlands}

\begin{abstract}
We provide a solution to the $v$-representability problem for a non-relativistic quantum many-particle system on a one-dimensional torus domain in terms of Sobolev spaces and their duals. Any one-particle density that is square-integrable, has a square-integrable weak derivative, and is gapped away from zero can be realized from the solution of a many-particle Schrödinger equation, with or without interactions, by choosing a corresponding external potential. This potential can contain a distributional contribution but still gives rise to a self-adjoint Hamiltonian. Importantly, this allows for a well-defined Kohn--Sham procedure but, on the other hand, invalidates the usual proof of the Hohenberg--Kohn theorem.
\end{abstract}
\maketitle

\tableofcontents

\section{Introduction}
\label{sec:intro}

The origin of the $v$-representability problem can be found in the seminal work of \citet[Footnote~12]{hohenberg-kohn1964} that established the field of density-functional theory (DFT). Therein, the universal density functional is only defined for densities with the property of $v$-representability. The term ``$v$-representability'' itself, according to \citet{levy1979universal}, is due to E.\ G.\ Larson at the Boulder Theoretical Chemistry Conference, June 1975. The same paper also presents the constrained-search approach as a way to extend the domain of the universal functional to all $N$-representable densities (those that come from $N$-particle wave functions). It was later realized by \citet{levy1982electron} and \citet{Lieb1983} that there exist densities that belong to a ground-state \emph{ensemble} of at least $3$-fold degeneracy, but not to any pure ground state. Also, it was found by \citet{ENGLISCH1983} that even this more general ensemble $v$-representability is not enough to generate all possible densities. Since it is also imaginable that Hamiltonians with and without interactions yield entirely different ground-state densities, we arrive at the following types of $v$-representability.

\begin{center}
    \setlength\tabcolsep{.5em}
    \begin{tabular}{c|c}
        non-interacting pure-state $v$-representability & non-interacting  ensemble $v$-representability \\ \hline
        interacting pure-state $v$-representability & interacting  ensemble $v$-representability
    \end{tabular}
\end{center}

This distinction is indeed important when it comes to the widely-applied method of \citet{KS1965}. The Kohn--Sham approach, where a non-interacting auxiliary system is employed to determine the density of an interacting physical system, relies on the assumption that the density is simultaneously interacting and non-interacting $v$-representable. And so the original $v$-representability problem is not solved just by the definition of constrained-search density functionals. Accordingly, the $v$-representability problem is one of the major open problems in the mathematical formulation of density-functional theory \cite{wrighton2023some,lewin-2023-in-DFT-book}. A special focus was given to the issue of $v$-representability also by Mathieu Lewin in \citet[4.5.2]{teale2022dft}.

In this work, we will always consider the ensemble case and at the same time allow for a whole class of different interaction operators $W$ in the Hamiltonian. Since this class includes $W=0$, the difference between interacting and non-interacting $v$-representability will not be relevant in our treatment and the presented solution applies simultaneously to both cases. The same general approach was followed in \citet{CCR1985}, where $v$-representability for all fermionic densities $0<\rho(x)<1$ on an (infinite) lattice was demonstrated. A simplification of the proof for finite lattices can be found in \citet{penz-DFT-graphs}.
\citet{Lammert2010} was able to provide $v$-representability for a coarse-grained formulation of DFT, and he also provided crucial insights into the continuum case~\cite{Lammert2007}. This work demonstrated that the universal density functional, as formalized in \citet{Lieb1983}, is not functionally differentiable. Since functional differentiability would imply $v$-representability, the counter-examples provided by \citet{Lammert2007} highlight mathematical obstacles that have to be overcome in order to guarantee that a density is indeed $v$-representable. In order to guarantee functional differentiability of a density $\rho(x)>0$, he suggested to work with the Sobolev space $H^2$ for densities instead of the Lebesgue spaces employed by \citet{Lieb1983}. Sobolev spaces will play a crucial role in this work.

Apart from this paper, notable progress has been achieved by \citet{aryasetiawan-stott1986,aryasetiawan-stott1988} and \citet{aryasetiawan-1989-personal-comm}, where $v$-representability is studied for a \emph{non-interacting} system on a one-dimensional bounded interval with zero boundary conditions. In this setting, $v$-representability can be achieved for any density that is strictly positive except at the boundary and has a finite derivative everywhere. A key ingredient is the transformation of the decoupled $N$-particle Schrödinger equation into $(N-1)$ coupled and non-linear differential equations that depend on the given density $\rho$ instead of the external potential $v$. While the issue of making sure that a solution also gives back $\rho$ as a ground-state density that appears in higher dimensions has been tackled in subsequent work by \citet{chen-stott1991-few-fermions,chen-stott1991-low-deg,chen-stott1993-v-rep}, the issue of solvability of the coupled non-linear differential equations was left open.
A recent work of \citet{garrigue2022building} approaches the issue of $v$-representability as an inverse problem with discretized potential space and also discusses the possibility of excited states.
Nevertheless, the original $v$-representability problem, especially in the case of interacting particles, must be still marked as basically unsolved. But before we move any further, we will give our formulation of the problem.

\vspace{.75em}
\noindent\fbox{\parbox{\textwidth}{
    \textbf{$v$-representability problem:} Find an explicit set of densities, preferably as large as possible, within the $N$-representable densities such that each element $\rho$ from this set is the density of an (ensemble) ground state of a self-adjoint Hamiltonian $H=H_0+V= -\frac{1}{2}\Delta + W + V$. While $H_0$, the kinetic and interaction contribution, is considered fixed, the external potential $V$ will depend on the chosen $\rho$. 
}}
\vspace{.75em}

From the above statement it follows that we also need to investigate for a chosen set of densities the corresponding set of one-body potentials $v$ and whether these potentials allow the definition of corresponding self-adjoint Hamiltonians.
Crucial insights in this regard can be obtained from an example presented in \citet{ENGLISCH1983}, which later has been discussed in more detail in \citet{CCR1985}. We will here quote some passages from the latter reference since their discussion already points towards the set of potentials that we will consider in this work. The example consists of the density 
\begin{equation}\label{eq:EE-dens}
    \rho(x)=(a+b|x|^{\alpha+1/2})^2
\end{equation}
with $a>b>0$, $1/2>\alpha>0$, that is defined locally around $x=0$, so it could be on the real line as well as on the one-dimensional torus. Now for a single particle with a ground-state wave function assumed to be real and having $\Psi>0$ (this is justified by a result in \citet[Th.~11.8]{Lieb-Loss-book}), we can take $\rho=\sqrt{\Psi}$. From the single-particle Schrödinger equation it then follows $v = (\Delta\sqrt{\rho}) / (2\sqrt{\rho})$, so we have immediate access to the representing potential. But it appears that
\begin{equation}\label{eq:EE-pot}
    \Delta\sqrt{\rho(x)} = b\left(\alpha+\frac{1}{2}\right) |x|^{\alpha-1/2} \delta(x) + b\left(\alpha^2-\frac{1}{4}\right) |x|^{\alpha-3/2},
\end{equation}
so the resulting potential is actually distribution-valued. \citet[p.~514f]{CCR1985} write:
``This $\rho$ is not the ground state of any Hamiltonian $-\Delta + v$, having a real-valued potential, although it may be the ground state of a Hamiltonian with a distributional potential.''
And further that the potential-density mapping ``may map very singular potentials into very smooth densities.'' They suggest that the problem might be solved by including generalized potentials, but that then ``much of HK [Hohenberg--Kohn] theory [...] may break down.''

We will take this example and their insights as the starting point here, indeed showing that by including such generalized potentials the $v$-representability problem can be solved. But quite contrary to their statement we would argue that only with this result, DFT has a real theoretical foundation. The example and our results below demonstrate that perfectly reasonable densities can only be realized with distributional potentials. What parts of the theory must then eventually be changed will be discussed in the conclusions (Section~\ref{sec:conclusions}). With this, we formulate the main result of this work.

\vspace{.75em}
\noindent\fbox{\parbox{\textwidth}{\vspace{-.75em}
\begin{theorem}\label{th:main}
    Define the following set of $v$-representable densities on the one-dimensional torus $\T$ and the dual space of (distributional) one-body potentials,
    \begin{align}
        \vrepset &= \{ \rho\in L^2(\T) \mid \nabla\rho\in L^2(\T), \smallint\rho=N, \forall x\in\T : \rho(x)>0 \}  \\
        \affspace^* &= \{ [v] = \{ v + c \mid c\in \R \} \mid v=f+ \nabla g \;\text{with}\; f,g \in L^2(\T) \}.
    \end{align}
    Then for every $\rho\in\vrepset$ there is an equivalence class $[v]\in\affspace^*$ with corresponding potential $V=\sum_{j=1}^N v(x_j)$ acting on wave functions, such that $\rho$ is the density of an (ensemble) ground state of the self-adjoint Hamiltonian $H=H_0+V = -\frac{1}{2}\Delta + W + V$.
\end{theorem}
\vspace{-.75em}}}
\vspace{.75em}

The interaction part $W$ is allowed to remain very general, it only needs to fulfil a property discussed in Corollary~\ref{cor:Hamiltonian}. Further details about the definitions in Theorem~\ref{th:main} will be given throughout the following sections. 
These discuss the basic setting of the $N$-particle Schrödinger problem on the one-dimensional torus in terms of quadratic forms (Section~\ref{sec:setting}), the considered spaces for densities and potentials (Section~\ref{sec:dens}), the convex formulation of DFT in the given setting (Section~\ref{sec:convex}), how representing potentials follow from a statement about subdifferentials (Section~\ref{sec:rep}), and the KLMN theorem for the generalised potentials (Section~\ref{sec:KLMN}). Finally, the proof of the main theorem above is given in the concluding Section~\ref{sec:conclusions}, building on all of the previous sections. The same section also discusses implications of our findings and future research directions. Appendix~\ref{app:kinetic-bounds-interactions} shows the critical property of `kinetic boundedness' (Definition~\ref{def:kinetic-bound}) for different classes of interactions.

\section{One-dimensional torus setting}
\label{sec:setting}

Our spatial domain is the one-dimensional torus $\T$, i.e., without loss of generality we take the interval $[0,1]$ where the points 0 and 1 are identified and we can recover any other one-dimensional periodic domain by a simple re-scaling. The reason to call it `one-dimensional torus' and not simply `ring' or `circle' is that this way we stay closer to the usual periodic setting in $d>1$ dimensions employed for extended solid-state systems.
Choosing a compact domain is critical for our proof strategy as it allows to infer that the density is gapped away from zero on the whole domain (cf.~Theorem~\ref{th:F-nonempty-subdiff}).
On it we consider a fixed number of $N$ spin-$\frac{1}{2}$ particles. The Hilbert space is then the $N$-fold tensor product of the spaces $L^2(\T)\otimes\C^2$ where only anti-symmetric wave functions are considered,
\begin{equation}
    \H = \bigl(L^2(\T)\otimes\C^2\bigr)^{\wedge N}.
\end{equation}
In what follows, we use the notations $\nabla_j$ and $\Delta_j$ for the first and second order (weak) derivatives w.r.t.\ $x_j$, even though the coordinate is one-dimensional.
For any wave function $\Psi \in \H$ we define the one-particle density on $\T$,
\begin{equation}\label{eq:def-dens}
    \rho_\Psi(x) = N \sum_{\sigma_1,\ldots,\sigma_N} \int \abs{\Psi (x\sigma_1, x_2\sigma_2, \dotsc, x_N\sigma_N)}^2 \d x_2 \dotsb \d x_N.
\end{equation}
Here, the sum over $\sigma_1,\dotsc,\sigma_N$ sums over all spin components, while all spatial particle coordinates except for one is integrated out as well. This density will be our main object of interest. Note that a $\Psi\in\H$ yields a $\rho_\Psi\in L^1(\T)$ like this, which does not allow for a point-wise evaluation $\rho_\Psi(x)$. This already points to the beneficial use of Sobolev spaces and indeed a density will later not only be pointwise well-defined but is even always a continuous function.

The Hamiltonian of the system is
\begin{equation}
    H = H_0 + V = -\frac{1}{2}\Delta + W + V = -\frac{1}{2}\sum_{j=1}^N \Delta_j + W + V,
\end{equation}
with the external potential $V=\sum_{j} v(x_j)$ being defined from a one-body potential $v$ and the interaction operator $W$ left as a general operator. This setting will soon be generalized. The expectation value of $V$ with respect to a wave function $\Psi$ is then given entirely in terms of the density as
\begin{equation}\label{eq:V-expectation}
    \langle \Psi, V \Psi \rangle = N \langle \Psi, v \Psi \rangle = \langle v,\rho_\Psi \rangle.
\end{equation}
Here, the first two angle brackets describe the inner product in $\H$ while the last one is the dual pairing between a density and a potential. Instead of linear operators we will consider the more general case of quadratic forms $V(\Phi,\Psi)$ and $W(\Phi,\Psi)$, usually still written as $\langle \Phi, V \Psi \rangle$ and $\langle \Phi, W \Psi \rangle$.

\begin{definition}[{\citet[Sec.~VIII.6]{reed-simon-1}}]
    A \emph{quadratic form} is a map $A : Q(A)\times Q(A) \to \C$, where $Q(A)$ is a dense subspace of $\H$ called the form domain. It is such that $A(\Phi,\Psi)$ is linear in the second argument and conjugate linear in the first. If $A(\Phi,\Psi) = \overline{A(\Psi,\Phi)}$ it is called \emph{symmetric} and if for all $\Psi\in Q(A)$ it holds $A(\Psi,\Psi) \geq 0$ then it is \emph{positive}.
\end{definition}

A self-adjoint or symmetric operator defines a symmetric quadratic form by $A(\Phi,\Psi) = \langle \Phi, A \Psi \rangle$ and equally a positive operator defines a positive quadratic form.
This makes it possible to take a distribution-valued $v$ as an external one-body potential (\citet[Sec.~VIII.6, Ex.~1]{reed-simon-1} also give such an example), and analogously to Eq.~\eqref{eq:V-expectation} we simply have
\begin{equation}
    V(\Psi,\Psi) = \langle v,\rho_\Psi \rangle.
\end{equation}
Now, the angle bracket means the application of the distribution $v$ on $\rho_\Psi$, which is the same as the dual pairing if a Banach space for $\rho_\Psi$ and its topological dual for $v$ are considered. The full quadratic form $V(\Phi,\Psi)$ can then be constructed with the polarization identity.
Adding such potentials to the Hamiltonian makes it a quadratic form as well. This means switching from the usual operator domain of the self-adjoint operator to the form domain \cite[Sec.~VIII.6]{reed-simon-1}. In case of the Laplacian $-\Delta = -\sum_{j} \Delta_j$, which is self-adjoint on all $\Psi\in\H$ that have $\Delta_j\Psi$ also $L^2$-integrable and where $\Psi$ and $\nabla_j\Psi$ are periodic, the resulting form domain is relatively easy to see. By symmetry of the wave function and by partial integration on the torus we have
\begin{equation}\label{eq:Laplace-expectation-val}
    \langle \Psi, -\Delta \Psi \rangle = \sum_{j=1}^N \langle \Psi, -\Delta_j \Psi \rangle = \sum_{j=1}^N \langle \nabla_j \Psi, \nabla_j \Psi \rangle,
\end{equation}
where the boundary terms cancel due to periodicity.
We note that $\langle \nabla_j \Psi, \nabla_j \Psi \rangle = \norm{\nabla_j\Psi}_{\H}^2$, so the form domain is constructed from $L^2(\T)$ functions that have an $L^2$-integrable weak derivative. This gives exactly the Sobolev space $H^1(\T)$ with norm $\norm{f}^2_{H^1} = \norm{f}^2_{L^2} + \norm{\nabla f}^2_{L^2}$. We find
\begin{equation}\label{eq:Laplace-form-domain}
    Q(-\Delta) = \bigl(H^1(\T)\otimes\C^2\bigr)^{\wedge N}.
\end{equation}
Under certain conditions, adding potentials in the form of quadratic forms to the Hamiltonian still allows to connect it to a self-adjoint operator. This is the content of the following classical theorem that is central to our work.

\begin{theorem}[KLMN theorem, \citet{reed-simon-2}, Th.~X.17]\label{th:KLMN}
    Let $A$ be a positive self-adjoint operator and $B$ a symmetric quadratic form on $Q(A)\times Q(A)$. Further, let $0\leq a < 1$ and $b\geq 0$ such that
    \begin{equation}
        |B(\Psi,\Psi)| \leq a\langle \Psi,A\Psi \rangle + b \langle \Psi,\Psi \rangle    
    \end{equation}
    for all $\Psi\in Q(A)$.
    Then there exists a unique self-adjoint operator $C$ with $Q(C)=Q(A)$, bounded below by $-b$, and 
    \begin{equation}
        \langle \Phi, C\Psi \rangle = \langle \Phi, A\Psi \rangle + B(\Phi,\Psi)
    \end{equation}
    for any $\Phi,\Psi\in Q(C)$.
\end{theorem}

The theorem also explains why we can write $B(\Phi,\Psi)$ and $\langle \Phi, B\Psi \rangle$ interchangeably, since we could just set $B = C-A$ with the operators $A,C$ from above.
We now get back to our case of $A=-\Delta$ with the form domain $Q(-\Delta)$ given in Eq.~\eqref{eq:Laplace-form-domain}. When normalization is added, this gives us our basic space for wave functions,
\begin{equation}
    \psiset = \{ \Psi \in (H^1(\T)\otimes\C^2)^{\wedge N} \mid \|\Psi\|_{\H}=1 \}.
\end{equation}
We also give the kinetic energy for a wave function $\Psi \in \psiset$,
\begin{equation}
    T(\Psi) = -\frac{1}{2} \langle \Psi, \Delta \Psi \rangle = \frac{1}{2}\sum_{j=1}^N \| \nabla_j \Psi \|_\H^2,
\end{equation}
which is always finite because of the restriction to the Sobolev space $H^1(\T)$ in $\psiset$. The following property will be demanded for the potentials that are added to the Hamiltonian in order to still have a self-adjoint operator by means of the subsequent corollary.

\begin{definition}\label{def:kinetic-bound}
    A quadratic form $B$ on $\psiset\times\psiset$ is called \emph{kinetically bounded}, if there are $a,b \geq 0$ such that for all $\Psi \in \psiset$ it holds $|B(\Psi,\Psi)| \leq a T(\Psi) + b$. The infimum of all such $a$ is called the \emph{relative kinetic bound}.
\end{definition}

\begin{corollary}\label{cor:Hamiltonian}
    Let $W$ be a quadratic form that is positive and kinetically bounded with relative kinetic bound $<1$. Then the interaction-only Hamiltonian $H_0$ is a self-adjoint operator with form domain $Q(H_0)=Q(-\Delta)$ that is positive and for all $\Phi,\Psi\in Q(-\Delta)$ it holds
    \begin{equation}
        \langle \Phi, H_0\Psi \rangle = -\frac{1}{2}\langle \Phi, \Delta\Psi \rangle + W(\Phi,\Psi).
    \end{equation}
    Take further $V$ as a quadratic form that is kinetically bounded with relative kinetic bound $<1$. Then the full Hamiltonian $H$ is a self-adjoint operator with form domain $Q(H)=Q(-\Delta)$ that is bounded below and for all $\Phi,\Psi\in Q(-\Delta)$ it holds
    \begin{equation}\label{eq:weak-SE}
        \langle \Phi, H\Psi \rangle = -\frac{1}{2}\langle \Phi, \Delta\Psi \rangle + W(\Phi,\Psi) + V(\Phi,\Psi).
    \end{equation}
\end{corollary}

The proof is just a double application of Theorem~\ref{th:KLMN} (KLMN theorem), first for $W$, which still leaves the combined operator positive, then for $V$. That $-\Delta$ is self-adjoint and positive itself was already noted after Eq.~\eqref{eq:Laplace-expectation-val}.
All considered potentials $V$ and $W$ will be such that the above corollary holds. Lemma~\ref{lem:v-kinetic-bound} shows that all potentials $V$ from the considered class have the desired property and Appendix~\ref{app:kinetic-bounds-interactions} shows it for a very broad variety of possible interactions $W$.

\section{Densities on the one-dimensional torus}
\label{sec:dens}

We define the set of \emph{physical} densities
\begin{equation}
    \densset = \{ \rho \in L^1(\T) \mid \nabla\sqrt{\rho} \in L^2(\T), \rho \geq 0, \smallint \rho = N \}.
\end{equation}
The special significance of this set is revealed by the following two results that connect it to wave functions and their kinetic-energy content.

\begin{lemma}\label{lem:psiset-implies-densset}
	Let $\Psi \in \psiset$. Then $\rho_\Psi\in\densset$ and it holds
    \begin{equation}\label{eq:psi-dens-estimate-1}
        \int \left|\nabla\sqrt{\rho_\Psi(x)}\right|^2 \d x \leq 2T(\Psi)
    \end{equation}
\end{lemma}
\begin{proof}
The properties $\rho_\Psi \in L^1(\T)$, $\rho_\Psi \geq 0$ and $\smallint \rho_\Psi = N$ follow straight from the definition of the density in Eq.~\eqref{eq:def-dens}. The proof for the estimate is then the same as in \citet[Th.~1.1]{Lieb1983}. With this estimate we also have $\nabla\sqrt{\rho_\Psi} \in L^2(\T)$ and consequently $\rho_\Psi \in \densset$.
\end{proof}
\begin{theorem}\label{th:N-rep}
    There exist two constants $C^T_1,C^T_2 > 0$ such that for all $\rho \in \densset$ there is a $\Psi \in \psiset$ that has $\rho_\Psi=\rho$ ($N$-representability) and the following estimate holds
    \begin{equation}\label{eq:psi-dens-estimate-2}
        T(\Psi)\leq C^T_1 + C^T_2 \int \left| \nabla \sqrt{\rho(x)}\right | ^2 \d x.
    \end{equation}
\end{theorem}
\begin{proof}
    We proceed as in \citet[Th.~1.2]{Lieb1983}. However, due to boundedness and periodicity of the domain there are some significant changes. As an ansatz for $N$-representability, we define $f(x)=(2\pi/N) \int_0^x \rho(s) \d s$ and for $k=0,\dotsc,N-1$ the orbitals
    \begin{equation}
        \phi_k(x)= \left(\frac{\rho(x)}{N} \right)^{1/2} \exp({\i k f(x)}). 
    \end{equation}
    One can check in \citet[Th.~1.2]{Lieb1983} that these orbitals are orthonormal and that the Slater-determinant of all $\phi_k$ gives indeed the density $\rho$. We focus here on estimating the kinetic energy that is the sum of all $\|\nabla\phi_k\|_{L^2}^2$. We have 
    \begin{equation}
    \label{eq:orbital-kin-energy}
        \norm{\nabla\phi_k}_{L^2}^2 = \int \abs{ \nabla \phi_k(x) }^2 \d x = \frac{1}{N} \int \left( \abs*{ \nabla \sqrt{\rho(x)}}^2+\left(\frac{2\pi k}{N}\right)^2 \rho(x)^3 \right) \d x.
    \end{equation}
    We introduce $g=\sqrt{\rho}$ and define 
    \begin{equation}
        A=\int \left|\nabla g(x)\right|^2  \d x =\int \left|\nabla \sqrt{\rho(x)}\right|^2 \d x.
    \end{equation}
    Since $\rho \geq 0$ and $\rho$ integrates to $N$, there needs to be one point $x_0\in\T$ where $g(x_0) \leq \sqrt{N}$ (a pointwise evaluation of the function is possible because of an embedding into the continuous functions, see Lemma~\ref{lem:embeddings}). Then we have
    \begin{equation}
        g(x)^2= g(x_0)^2 + 2 \int_{x_0}^x g(y) \nabla g(y) \d y\leq N + 2\left[\int g(y)^2 \d y \right]^{1/2}\left[\int |\nabla g(y)|^2 \d y \right]^{1/2}= N+2\sqrt{NA},
    \end{equation}
    where we used Hölder's inequality, extended the integration over the full range of the torus, and that $g(x_0) \leq \sqrt{N}$ as well as $g^2=\rho$. Using $2\sqrt{NA} \leq N + A$ it then follows $g(x)^4 \leq 3N^2+6NA$. The second term in \eqref{eq:orbital-kin-energy}
    can thus be bounded from above by
    \begin{equation}
        \frac{(2\pi k)^2}{N^3} \int \rho(x)^3 \d x = \frac{(2\pi k)^2}{N^3} \int \rho(x) g(x)^4 \d x \leq \left(\frac{2\pi k}{N}\right)^2 (3N^2+6NA),
    \end{equation}
    and so the statement of Eq.~\eqref{eq:psi-dens-estimate-2} follows from summing $k$ from $0$ to $N-1$.
\end{proof}

Although the results above show that $\densset$ gives an exhaustive set of physical densities, i.e., every such density has a representing $N$-particle wave function and every wave function in $\psiset$ gives a density in $\densset$, it is still unpractical for our purposes, because by the condition $\nabla\sqrt{\rho} \in L^2(\T)$ it does not yield a \emph{linear} space. Our main space of interest is instead the \emph{affine space} of codimension 1 (having one linear constraint, the normalization of the density)
\begin{equation}
    \affspace = \{ \rho \in H^1(\T) \mid \smallint \rho = N \}.
\end{equation}

Remember that a normed space $X \subseteq Y$ is called \emph{continuously embedded} in $Y$, written $X\hookrightarrow Y$, if there is a constant $C$ such that for all $\xi \in X$ it holds $\|\xi\|_Y \leq C \|\xi\|_X$. The following lemma shows a whole chain of (continuous) embeddings.

\begin{lemma}\label{lem:embeddings}
    $\densset \subset \affspace \subset H^{1}(\T)\hookrightarrow C^0(\T) \hookrightarrow L^\infty(\T) \hookrightarrow L^3(\T) \hookrightarrow L^2(\T) \hookrightarrow L^1(\T)$.
\end{lemma}

\begin{proof}
    The continuous (even compact) embedding $H^{1}(\T)\hookrightarrow C^0(\T)$ is a consequence of the Sobolev embedding theorem \cite[Th.~8.8]{Brezis-book}. Here the space of continuous functions is equipped with the usual maximum norm, so the next continuous embedding $C^0(\T) \hookrightarrow L^\infty(\T)$ follows directly. The remaining embeddings to the right hold since $\T$ is a finite domain. Since all $\rho \in \densset$ have $\sqrt{\rho} \in L^2(\T)$ and $\nabla\sqrt{\rho} \in L^2(\T)$, equivalent to $\sqrt{\rho} \in H^1(\T)$, we have for such densities that $\sqrt{\rho} \in C^0(\T)$ and consequently also $\rho$ is continuous. This allows the following estimate,
    \begin{equation}\label{eq:lem:embeddings}
        \|\nabla\rho\|_{L^2}^2 = \int |\nabla\rho(x)|^2 \d x = 4 \int \rho(x) \left|\nabla\sqrt{\rho(x)}\right|^2 \d x \leq 4 \| \rho \|_{L^\infty} \|\nabla\sqrt{\rho}\|_{L^2}^2,
    \end{equation}
    which shows the first inclusion $\densset \subset \affspace$. Finally, $\affspace \subset H^{1}(\T)$ trivially holds since $\affspace$ is an affine space in $H^1(\T)$.
\end{proof}

For later purposes we will also study the dual space of $\affspace$ denoted as $\affspace^*$, later to be identified with the space of external potentials. The dual space of the Sobolev space $H^1(\T)$, denoted $H^{-1}(\T)$, is composed of all distributions $v:H^1(\T)\to\R$ that can be written as $v=f+\nabla g$ with $f,g\in L^2(\T)$. This dual space is endowed with the norm $\| v \|^2_{H^{-1}}= \min \{ \| f \|^2_{L^2} + \| g \|^2_{L^2} \mid f,g \in L^2(\T), v=f+ \nabla g\}$ \cite[Th.~3.9]{adams-book}. An element $v\in H^{-1}(\T)$ acts on a function $\varphi \in H^1(\T)$ as $v(\varphi) = \langle v,\varphi \rangle = \langle f,\varphi \rangle - \langle g,\nabla \varphi \rangle$. Additionally to $H^{1}(\T)$, the affine space $\affspace$ is defined by the normalization of its elements to the particle number $N$, so if two elements of the dual space $\affspace^*$ just differ by a constant, $v-v'=c$, then their difference acts as $\langle v-v',\rho \rangle = c N$ on any arbitrary $\rho\in \affspace$. Consequently, $v$ and $v'$ cannot be distinguished and the dual space $\affspace^*$ corresponds to $H^{-1}(\T)$ modulo the addition of constant functions. We thus have $\affspace^*$ given by equivalence classes,
\begin{equation}\label{eq:v-dual-space}
    \affspace^* = \{ [v] = \{ v + c \mid c\in \R \} \mid v=f+ \nabla g \;\text{with}\; f,g \in L^2(\T) \}.
\end{equation}
In what follows, we commonly write $v$ instead of $[v]$ and by this mean any representative of the equivalence class. Physically, if $v$ is a potential, this form of equivalence is well known, since the addition of a constant only shifts the energy but does not change any ground-state properties. This equivalence is due to the \emph{gauge freedom} that we have when mathematically representing an electromagnetic field in terms of potentials.

\begin{example}\label{ex:delta}
The delta distribution $\delta : H^1(\T) \to \R$ is an element of $H^{-1}(\T)$. For this we have to show that there are $f,g\in L^2(\T)$ such that $\delta = f + \nabla g$. Take $g(x)=-x$ and $f(x) = 1$ on $[0,1) \simeq \T$. Then $\delta(\varphi) = \langle \delta,\varphi \rangle = \langle f,\varphi \rangle - \langle g,\nabla\varphi \rangle$, or written out with integrals,
\begin{equation}
    \delta(\varphi) = \int_0^1 \varphi(x) \d x + \int_0^1 x\nabla\varphi(x)\d x.
\end{equation}
Then, by partial integration in the second integral we get
\begin{equation}
    \delta(\varphi) = \int_0^1 \varphi(x) \d x - \int_0^1 \varphi(x)\d x + x\varphi(x)\Big|_0^1 = \varphi(1).
\end{equation}
By Lemma~\ref{lem:embeddings} the $\varphi \in H^1(\T)$ is continuous and it is also periodic, thus $\varphi(1) = \varphi(0)$, and we get $\delta(\varphi) = \varphi(0)$, precisely the action of the delta distribution.
\end{example}

\begin{example}\label{ex:potentials-act}
How does an element $v=f+\nabla g\in\affspace^*$ act as a potential on wave functions or directly on densities in the dual pairing $\langle v,\rho \rangle$ that gives the potential energy? We have $V=\sum_k v(x_k)$ as the associated operator and thus for $V\Psi$ in a distributional sense for any test function $\Phi\in (C^\infty(\T) \otimes \C^2)^{\otimes N}$ it holds
\begin{equation}
    \langle \Phi, V\Psi \rangle = \sum_{k=1}^N \left( \langle \Phi, f(x_k)\Psi \rangle - \langle \Phi, g(x_k)\nabla\Psi \rangle - \langle \nabla\Phi, g(x_k)\Psi \rangle \right).
\end{equation}
This has a certain resemblance to the action of the Hamiltonian in a magnetic Schrödinger equation where terms like $\i\mathbf{A}\cdot\nabla \Psi$ appear. But the complex prefactor $\i$ and the missing $\langle \nabla\Phi,\cdot\rangle$ term amount to a critical difference. In the dual pairing the effect of the distributional potential is seen way simpler as
\begin{equation}
    \langle v,\rho \rangle = \langle f,\rho \rangle - \langle g,\nabla\rho \rangle = \int_0^1 f(x)\rho(x) \d x - \int_0^1 g(x)\nabla\rho(x)\d x,
\end{equation}
thus involving the density in a semi-local fashion.
\end{example}

\section{Convex formulation of density-functional theory}
\label{sec:convex}

The ground-state problem consists of finding the eigenstate (or the whole eigenspace in the case of degeneracy) corresponding to the lowest eigenvalue of $H_0 + V$. Here, we reformulate this eigenproblem into a more general variational problem, by seeking the infimum of the spectrum. For this, we rely on the variational principle for the lower bound of the spectrum of a self-adjoint operator \cite[Th.~XIII.1 for $n=1$]{reed-simon-4},

\begin{equation}\label{eq:E-def}
    E(v) \coloneqq \inf \sigma(H) = \inf \sigma(H_0+V) = \inf_{\Psi\in\psiset} \langle \Psi, (H_0+V) \Psi \rangle < \infty.
\end{equation}
The notation $E(v)$, with the one-body potential $v$ as an argument, is due to the fact that this $v$ will be the only variable component of the system, while $H_0$ is considered fixed.
The definition with an infimum directly implies for $\lambda \in [0,1]$ that
\begin{equation}\begin{aligned}
    E(\lambda v + (1-\lambda) v') &= \inf_{\Psi\in\psiset} \langle \Psi, (\lambda H_0 + (1-\lambda) H_0 +\lambda V + (1-\lambda) V') \Psi \rangle \\
    &\geq \lambda \inf_{\Psi\in\psiset} \langle \Psi, (H_0+V) \Psi \rangle + (1-\lambda) \inf_{\Psi\in\psiset} \langle \Psi, (H_0+V') \Psi \rangle = \lambda E(v) + (1-\lambda) E(v'),
\end{aligned}\end{equation}
so the functional $E$ is concave.
It is possible to separate the variation over wave functions in Eq.~\eqref{eq:E-def} into first a variation over all densities that follow from wave functions in $\psiset$, i.e., the set $\densset$ by Lemma~\ref{lem:psiset-implies-densset}, and then vary over all wave functions that give such a density, denoted as $\Psi\mapsto\rho$. We thus have
\begin{equation}\label{eq:E-from-tilde-F}
    E(v) = \inf_{\rho\in\densset} \{ \tilde F(\rho) + \langle v,\rho \rangle \}, \quad\quad \tilde F(\rho) \coloneqq \inf_{\Psi\mapsto\rho} \langle \Psi, H_0 \Psi \rangle.
\end{equation}
$\tilde F$ is called the \emph{pure-state constrained-search functional}~\cite{levy1979universal,Lieb1983}.
Note that any $\rho \in \affspace \setminus \densset$ is either not a (non-negative) density or it is connected to infinite kinetic energy by the estimate Eq.~\eqref{eq:psi-dens-estimate-1} in Lemma~\ref{lem:psiset-implies-densset}. In both cases we set $\tilde F(\rho)=\infty$. Consequently, the variation in Eq.~\eqref{eq:E-from-tilde-F} can be extended from $\densset$ to $\affspace$ without changing the result and we get
\begin{equation}
    E(v) = \inf_{\rho\in\affspace} \{ \tilde F(\rho) + \langle v,\rho \rangle \}.
\end{equation}
Next, define a functional $F$ on $\affspace$ by means of convex conjugation (Legendre--Fenchel transformation),
\begin{equation}\label{eq:F-def}
    F(\rho) \coloneqq \sup_{v\in\affspace^*} \{ E(v) - \langle v,\rho \rangle \}.
\end{equation}
This gives $F$ as the biconjugate of $\tilde F$ and as such $F$ is convex and lower-semicontinuous (even weakly lower-semicontinuous) and $F\leq  \tilde F$ \cite[Prop.~2.19]{Barbu-Precupanu}. We can also arrive at $F$ by taking the lower-semicontinuous convex envelope of $\tilde F$ \cite[Cor.~2.23]{Barbu-Precupanu} or by generalizing the constrained search in Eq.~\eqref{eq:E-from-tilde-F} from wave functions to density matrices~\cite{valone1980b,Lieb1983}. This fact will later prove critical for showing the final step for $v$-representability, in order to connect the density-functional formulation back to quantum-mechanical states.

\begin{theorem}\label{th:FDM}
The functional $F$ on $\densset$ is equal to the constrained search over density matrices, i.e., for all $\rho\in\densset$ it holds
\begin{equation}
    F(\rho) = \FDM(\rho) \coloneqq \inf_{\Gamma\mapsto\rho}\{ \Tr \Gamma H_0 \}.
\end{equation}
Here, all considered density matrices are (possibly infinite) convex sums $\Gamma=\sum_k \lambda_k \langle \Psi_k,\cdot \rangle \Psi_k$, $\lambda_k\geq 0$, $\sum_k\lambda_k=1$, with $\Psi_k \in \psiset$ all orthogonal. The density is then also a convex sum, $\rho=\sum_k\lambda_k\rho_k$, with $\Psi_k \mapsto \rho_k$.
\end{theorem}

\begin{proof}
For the first part, $F=\FDM$, remember that \citet[Th.~4.3]{Lieb1983} shows that $\FDM=\FLieb$ (the functional corresponding to Eq.~\eqref{eq:F-def} but with $v\in L^{3/2}+L^\infty$),  with the latter also equal to the lower-semicontinuous convex envelope of the same $\tilde F$ but w.r.t.\ the $L^1\cap L^3$-topology. Now with the embeddings from Lemma~\ref{lem:embeddings} we have that the topology of $H^1(\T)$ is finer than those of $L^1(\T)$ and $L^3(\T)$ and thus also of $L^1(\T)\cap L^3(\T)$. This implies that a function that is continuous (or lower-semicontinuous) w.r.t.\ the $L^1\cap L^3$-norm is also continuous (or lower-semicontinuous) w.r.t.\ the $H^1$-norm. In other words, the set of all functions continuous (or lower-semicontinuous) w.r.t.\ the $H^1$-norm is larger than the set of functions continuous (or lower-semicontinuous) w.r.t.\ the coarser topology. We conclude that the convex envelope (as the pointwise supremum over all convex, lower-semicontinuous functions that lie below) w.r.t.\ the $H^1$-topology is always larger or equal to the convex envelope w.r.t.\ the $L^1\cap L^3$-topology, thus $F \geq \FLieb = \FDM$. The inequality in the other direction follows easily by noting that the ground-state energy can be equally well defined by considering density matrices instead of pure states,
\begin{equation}
    E(v) = \inf_{\rho\in\densset} \{ \FDM(\rho) + \langle v,\rho \rangle \},
\end{equation}
and so for all $\rho\in\densset$ we have $E(v)\leq \FDM(\rho) + \langle v,\rho \rangle$. But by definition
\begin{equation}
    F(\rho) = \sup_{v\in\affspace^*} \{ E(v) - \langle v,\rho \rangle \} \leq \sup_{v\in\affspace^*} \FDM(\rho) = \FDM(\rho),
\end{equation}
which shows $F=\FDM=\FLieb$ on $\densset$.\\
That the density matrix is always a sum of orthogonal and normalized states $\Psi_k\in\H$ is a standard result \cite[Sec.~4.B]{Lieb1983}. That they are all $\Psi_k\in\psiset$ is clear since the kinetic energy expectation value would be infinite otherwise.
\end{proof}


Note that something critical happened in Eq.~\eqref{eq:F-def} above. The potential $v$ is now from the space $\affspace^*$ that could in general be different from the class of allowed potentials as demanded by Corollary~\ref{cor:Hamiltonian}.  This means one potentially loses the connection to a self-adjoint operator formalism and switches entirely to a convex analysis treatment of the ground-state problem. In Section~\ref{sec:KLMN} we will take special care of this and make the connection back to a self-adjoint Hamiltonian with potential $v\in\affspace^*$.

Finally, by biconjugation it is also possible to go back to $E$ from $F$ with the opposite convex conjugation,
\begin{equation}\label{eq:E-from-F}
    E(v) = \inf_{\rho\in\affspace} \{ F(\rho) + \langle v,\rho \rangle \}.
\end{equation}
This shows the common feature of all formulations of DFT: The previous ground-state problem given by a \emph{linear} many-particle Schrödinger equation with its high-dimensional configuration space is translated into a variational problem incorporating a \emph{non-linear} convex functional on just the one-particle density space.
But note that there is a critical difference between this way of defining the so-called \emph{universal density functional} $F$ and how this is done in most of the DFT literature. This difference lies in our restriction to the affine space $\affspace$ where the normalization to the particle number is already included. If $F$ is defined on a larger space like $H^1(\T)$ then arbitrarily close to any physical density there lie elements that are not even $N$-representable. This can be avoided if we limit ourselves to the affine space $\affspace$ as we will see in the next section. One important ingredient is the following bound on $F(\rho)$ for $\rho\in\densset$.

\begin{lemma}\label{lem:F-bound}
    There exist two constants $C_1^F, C_2^F > 0$ such that for all $\rho \in \densset$ it holds
    \begin{equation}
        \frac{1}{2}\| \nabla \sqrt{\rho} \|_{L^2}^2 \leq F(\rho) \leq C_1^F + C_2^F \| \nabla \sqrt{\rho} \|_{L^2}^2.    
    \end{equation}
    If $\rho \notin \densset$ then $F(\rho)= \infty$.
\end{lemma}

\begin{proof}
Define $G(\rho) = \frac{1}{2}\|\nabla \sqrt{\rho} \|_{L^2}^2$ for $\rho\in\densset$ and $G(\rho)=\infty$ otherwise.
Then for the lower bound use that for all $\rho\in L^1\cap L^3$ \citet[Th.~3.8]{Lieb1983} already showed that $G(\rho) \leq F(\rho)$. This already implies that for $\rho\notin\densset$ we have $G(\rho)=F(\rho)=\infty$.\\
For the upper bound we employ Eq.~\eqref{eq:psi-dens-estimate-2} of Theorem~\ref{th:N-rep} that states that for $\rho \in \densset$ there is a $\Psi\in\psiset$ with $\rho_\Psi=\rho$ and $T(\Psi) \leq  C_1^T + C_2^T \| \nabla\sqrt{\rho} \|_{L^2}^2 $. For the interaction term we have due to the kinetic boundedness according to Definition~\ref{def:kinetic-bound} $|\langle \Psi,W\Psi \rangle| \leq aT(\Psi)+b$. Both estimates together give an upper bound for $F(\rho)$,
\begin{equation}
    F(\rho) \leq \tilde F(\rho) \leq \langle \Psi, H_0 \Psi \rangle \leq T(\Psi) + \langle \Psi, W \Psi \rangle \leq  C_1^F + C_2^F \| \nabla\sqrt{\rho} \|_{L^2}^2, 
\end{equation}
which concludes the proof.
\end{proof}

\section{\texorpdfstring{$v$}{v}-representing potentials from non-empty subdifferentials}
\label{sec:rep}

In the previous section we found the ground-state energy and density from the variational problem
\begin{equation}\label{eq:sec:rep:E-from-F}
    E(v) = \inf_{\rho\in\affspace} \{ F(\rho) + \langle v,\rho \rangle \},
\end{equation}
where $F$ is a convex and lower-semicontinuous functional on $\affspace$. A convex functional allows the definition of the \emph{subdifferential} at $\rho\in\affspace$ as a set in the dual space $\affspace^*$,
\begin{equation}
    \partial F(\rho) = \{ u\in \affspace^* \mid \forall \rho'\in\affspace : F(\rho)-F(\rho') \leq \langle u,\rho-\rho' \rangle \}.
\end{equation}
An element $u \in \partial F(\rho)$ is then called a \emph{subgradient}. The subdifferential may be empty. If $F$ is (Gateaux) differentiable at some point $\rho\in\affspace$ then $\partial F(\rho)$ contains just a single element that gives exactly the derivative. Note that while the gradient is a local property for general functionals, the subgradient contains \emph{global} information of the convex functional through the inequality $F(\rho) - \langle u,\rho \rangle \leq F(\rho') - \langle u,\rho'\rangle$ which holds for any $\rho'\in\affspace$ if $u \in \partial F(\rho)$. If we just replace $u=-v$ then we see that this is exactly the condition for minimizing the variational problem of Eq.~\eqref{eq:sec:rep:E-from-F}.
This means the variational problem of finding the density that minimizes the functional $\rho \mapsto F(\rho) + \langle v,\rho \rangle$ can equivalently be replaced by the condition $-v \in \partial F(\rho)$. This in turn means that if $\partial F(\rho)$ is non-empty, we know that at least one $v\in\affspace^*$ exists for which $\rho$ is a minimizer in Eq.~\eqref{eq:sec:rep:E-from-F}.

The critical ingredients in order to get a non-empty subdifferential are then the following definition and two lemmata.

\begin{definition}\label{def:eff-domain}
    The \emph{effective domain} of a convex functional $F$ on $\affspace$ consists of all points $\rho\in\affspace$ where $F(\rho)<\infty$.
\end{definition}

\begin{lemma}[\citet{Barbu-Precupanu}, Th.~2.14]\label{lem:F-continuous}
    Let $F$ be a convex functional on $\affspace$ then $F$ is continuous on the whole interior of its effective domain if and only if $F$ is uniformly bounded from above on a neighborhood of any interior point of its effective domain.
\end{lemma}

\begin{lemma}[\citet{Barbu-Precupanu}, Prop.~2.36]\label{lem:F-continuous-nonempty-subdiff}
    If a convex functional $F$ on $\affspace$ is continuous at some $\rho\in\affspace$ then it has a non-empty subdifferential at this point.
\end{lemma}

The strategy is then clear. We need to identify a set in the interior of the effective domain of $F$ and show the necessary bound from above. For this we define a set of densities in $\affspace$ that are gapped away from zero,
\begin{equation}
    \vrepset = \{ \rho\in\affspace \mid \forall x\in\T : \rho(x)>0 \}.
\end{equation}

\begin{lemma}\label{lem:open}
    The set $\vrepset$ is open in $\affspace$.
\end{lemma}

\begin{proof}
    Note that the set $\{ f\in L^\infty(\T) \mid \forall x\in\T : f(x)>0 \}$ is open in $L^\infty(\T)$. Because of the continuous embedding $H^1(\T) \hookrightarrow L^\infty(\T)$ from Lemma~\ref{lem:embeddings} the set $\{ f\in H^1(\T) \mid \forall x\in\T : f(x)>0 \}$ is then open in $H^1(\T)$. Hence the restriction to the affine subspace $\affspace$ is also open.
\end{proof}

We then have the desired result that guarantees us a non-empty subdifferential.

\begin{theorem}\label{th:F-nonempty-subdiff}
    For any $\rho \in \vrepset$ it holds $\rho\in\densset$ and the subdifferential $\partial F(\rho) \subset \affspace^*$ is non-empty.
\end{theorem}

\begin{proof}
    Any $\rho \in \vrepset$ is continuous on $\T$ by Lemma~\ref{lem:embeddings}, so we can choose a minimum $\eta = \min_x \{\rho(x)\} > 0$.
    We then have the estimate
    \begin{equation}\label{eq:F-nonempty-subdiff-estimate}
        \|\nabla\sqrt{\rho}\|_{L^2}^2 = \int \left|\nabla\sqrt{\rho(x)}\right|^2 \d x = \frac{1}{4}\int \frac{|\nabla\rho(x)|^2}{\rho(x)} \d x \leq \frac{1}{4\eta}\int |\nabla\rho(x)|^2 \d x = \frac{1}{4\eta} \|\nabla\rho\|_{L^2}^2 \leq \frac{1}{4\eta}\| \rho \|_{H^1}^2. 
    \end{equation}
    This inequality already establishes that $\rho\in\densset$.
    Together with Lemma~\ref{lem:F-bound} this also means that $F$ is bounded above at any such $\rho \in \vrepset$. Thus every $\rho \in \vrepset$ belongs to the effective domain of $F$ and since by Lemma~\ref{lem:open} $\vrepset$ is open, the whole set must belong to the interior of the effective domain of $F$.
    By the continuous embedding $H^{1}(\T)\hookrightarrow L^\infty(\T)$ from Lemma~\ref{lem:embeddings}, we can choose a neighborhood for any $\rho \in \vrepset$ that is contained within an $\varepsilon$-ball around $\rho$ in $L^\infty$-norm. Let $\varepsilon < \eta$, then an estimate like in Eq.~\eqref{eq:F-nonempty-subdiff-estimate} holds uniformly on the whole neighborhood of $\rho \in \vrepset$. This carries over to a uniform bound for $F(\rho)$.
    Thus, Lemma~\ref{lem:F-continuous} becomes applicable and $F$ is continuous on $\vrepset$. Lemma~\ref{lem:F-continuous-nonempty-subdiff} then concludes the proof.
\end{proof}

This proves that every $\rho\in\vrepset$ is `representable' by a $v\in -\partial F(\rho) \subset \affspace^*$ in the sense that $E(v) = F(\rho) + \langle v,\rho \rangle$, i.e., $\rho$ is a minimizer in the variational problem of Eq.~\eqref{eq:sec:rep:E-from-F}.
Yet, we do not consider this already as sufficient for showing real $v$-representability, since it should also be guaranteed that this $v$ can be the external potential of a valid Hamiltonian.

\section{KLMN theorem for distributional potentials}
\label{sec:KLMN}

In order to show that $H = H_0 + V$ with $V$ defined by a $v\in\affspace^*$ as described in Section~\ref{sec:setting} is connected to a valid self-adjoint operator we have to look to Corollary~\ref{cor:Hamiltonian}. The critical condition on $V$ for this corollary to hold is that it is kinetically bounded with relative kinetic bound $<1$. Here we will show that the relative kinetic bound is $0$ (which does not mean that the $T(\Psi)$ in the estimate vanishes, just that the prefactor can be made arbitrarily small).

\begin{lemma}\label{lem:v-kinetic-bound}
    Let $v\in\affspace^*$. Then, $V$ is kinetically bounded on $\psiset\times\psiset$ with relative kinetic bound $0$.
\end{lemma}

\begin{proof}
This proof is inspired by the treatment of distributional potentials in the context of the KLMN theorem by \citet{HERCZYNSKI1989-KLMN}.
According to Eq.~\eqref{eq:V-expectation} and Definition~\ref{def:kinetic-bound} we have to test $| \langle v,\rho_\Psi \rangle | \leq \eps T(\Psi) + b_\eps$ for all $\Psi \in \psiset$ where $\eps$ can be taken arbitrarily small. Taking $\eps$ smaller will typically mean that the other bound $b_\eps$ grows. According to Eq.~\eqref{eq:v-dual-space} any $v\in\affspace^*$ can we written as $v=f+\nabla g$ where $f,g\in L^2(\T)$. We use the triangle inequality,
\begin{equation}
    | \langle v,\rho_\Psi \rangle | \leq | \langle f,\rho_\Psi \rangle | + | \langle \nabla g,\rho_\Psi \rangle |,
\end{equation}
and we will continue with each term individually.
We choose a sequence $\{g_n\}_n$ in $C^\infty(\T)$ such that $\|g-g_n\|_{L^2}\to 0$. We take the supremum of $\nabla g_n$ out of the inner product and get by the definition of the distributional part of the potential that
\begin{equation}\label{eq:nabla-g-estimate}
    | \langle \nabla g,\rho_\Psi \rangle | \leq | \langle \nabla g_n,\rho_\Psi \rangle | + | \langle \nabla (g-g_n),\rho_\Psi \rangle | \leq N \|\nabla g_n\|_{L^\infty} + | \langle g-g_n,\nabla \rho_\Psi \rangle | \leq N \|\nabla g_n\|_{L^\infty} + \|g-g_n\|_{L^2} \|\nabla\rho_\Psi\|_{L^2}.
\end{equation}
From Eq.~\eqref{eq:lem:embeddings} we have $\|\nabla\rho_\Psi\|_{L^2} \leq 2 \|\sqrt{\rho_\Psi}\|_{L^\infty} \|\nabla\sqrt{\rho_\Psi}\|_{L^2}$ and from the continuous embedding $H^1(\T) \hookrightarrow L^\infty(\T)$ of Lemma~\ref{lem:embeddings} that $\|\sqrt{\rho_\Psi}\|_{L^\infty} \leq C \|\sqrt{\rho_\Psi}\|_{H^1}$. Since $\|\nabla\sqrt{\rho_\Psi}\|_{L^2} \leq \|\sqrt{\rho_\Psi}\|_{H^1}$ this combines to $\|\nabla\rho_\Psi\|_{L^2} \leq 2 C \|\sqrt{\rho_\Psi}\|_{H^1}^2 = 2C(N+\|\nabla\sqrt{\rho_\Psi}\|_{L^2}^2) \leq 2C(N+2T(\Psi))$ by Eq.~\eqref{eq:psi-dens-estimate-1} in Lemma~\ref{lem:psiset-implies-densset}. So the estimate is now
\begin{equation}
    | \langle \nabla g,\rho_\Psi \rangle | \leq N \|\nabla g_n\|_{L^\infty} + 2C \|g-g_n\|_{L^2} (N+2T(\Psi)).
\end{equation}
For $| \langle f,\rho_\Psi \rangle |$ the estimate is simpler since no partial integration is necessary. Again, we introduce a sequence $\{f_n\}_n$ in $C^\infty(\T)$ such that $\|f-f_n\|_{L^2}\to 0$ and have like in Eq.~\eqref{eq:nabla-g-estimate} that
\begin{equation}
    | \langle f,\rho_\Psi \rangle | \leq N \|f_n\|_{L^\infty} + \|f-f_n\|_{L^2} \|\rho_\Psi\|_{L^2}. 
\end{equation}
We then use $\|\rho_\Psi\|_{L^2} \leq \|\sqrt{\rho_\Psi}\|_{L^\infty} \|\sqrt{\rho_\Psi}\|_{L^2}$ and from Lemma~\ref{lem:embeddings} also the continuous embedding $H^1(\T) \hookrightarrow L^2(\T)$ to get $\|\rho_\Psi\|_{L^2} \leq C' \|\sqrt{\rho_\Psi}\|_{H^1}^2 \leq C'(N+2T(\Psi))$. Then the second estimate is
\begin{equation}
    | \langle f,\rho_\Psi \rangle | \leq N \|f_n\|_{L^\infty} + C' \|f-f_n\|_{L^2} (N+2T(\Psi))
\end{equation}
and we can combine this to
\begin{equation}
    | \langle v,\rho_\Psi \rangle | \leq N \left(\|f_n\|_{L^\infty} + \|\nabla g_n\|_{L^\infty} \right) + \left(2C \|g-g_n\|_{L^2} + C' \|f-f_n\|_{L^2} \right) (N+2T(\Psi)).
\end{equation}
But by increasing the index $n$ the $\|f-f_n\|_{L^2}$ and $\|g-g_n\|_{L^2}$ can be taken arbitrarily small such that the whole prefactor for $T(\Psi)$ drops below any $\eps>0$, while the remaining terms increase to some large but finite value.
\end{proof}

This result makes Theorem~\ref{th:KLMN} (KLMN theorem) applicable to potentials $v\in\affspace^*$. We will go one step further with the next theorem, showing also the corresponding Hamiltonian always allows for a ground state in the appropriate sense. As a preparatory step we show the following equivalence of norms.

\begin{lemma}\label{lem:norm-equiv}
    Let $v\in\affspace^*$ then the corresponding potential form $V$ defines a self-adjoint Hamiltonian $H=H_0+V$ that is bounded below. With an appropriate shift of the potential $v$ by adding a constant, the square root of $\langle \Psi,H\Psi \rangle$ is equivalent to the $H^1$-norm for all $\Psi\in\psiset$.
\end{lemma}

\begin{proof}
    That any $v\in\affspace^*$ defines a self-adjoint Hamiltonian that is bounded below is proven if Lemma~\ref{lem:v-kinetic-bound} and Corollary~\ref{cor:Hamiltonian} are combined.
    By Lemma~\ref{lem:v-kinetic-bound} for $V$ and by assumption for $W$, both quadratic forms have relative kinetic bounds $<1$, so we can find $0\leq a_V,a_W <1$ and $b_V,b_W\geq 0$ like in Definition~\ref{def:kinetic-bound} where $a_V$ can even be taken arbitrarily small, so $a_V+a_W<1$.
    Now reordering the expectation value
    $\langle \Psi,H\Psi \rangle = T(\Psi) + \langle \Psi,W\Psi \rangle + \langle \Psi,V\Psi \rangle$
    gives the estimate
    \begin{equation}
        T(\Psi) \leq \langle \Psi,H\Psi \rangle + |\langle \Psi,W\Psi \rangle| + |\langle \Psi,V\Psi \rangle| \leq \langle \Psi,H\Psi \rangle + aT(\Psi) + b,
    \end{equation}
    where we introduced $0\leq a = a_V+a_W < 1$ and $b = b_V+b_W \geq 0$. Thus we have
    \begin{equation}
        (1-a) T(\Psi) - b \leq \langle \Psi,H\Psi \rangle
    \end{equation}
    as a lower bound for $\langle \Psi,H\Psi \rangle$.
    Next, remember that a potential $v\in\affspace^*$ is only defined up to an additive constant as an equivalence class, so we are actually free to choose an arbitrary additive constant for $H$. Let $H' = H + (1-a)/2+b$ then
    \begin{equation}
        \frac{1-a}{2} (2T(\Psi) +1) \leq \langle \Psi,H'\Psi \rangle.
    \end{equation}
    Since $2T(\Psi) + 1$ is just the square of the $H^1$-norm for a (normalized) $\Psi\in\psiset$, this gives the lower bound for the equivalence. The upper bound is seen more directly from the same kinetic bounds for the expectation values of $V$ and $W$.
    \begin{equation}
        \langle \Psi,H'\Psi \rangle = \langle \Psi,H\Psi \rangle + \frac{1-a}{2}+b \leq (1+a)T(\Psi) + \frac{1-a}{2} + 2b \leq \max\left\{ \frac{1+a}{2}, \frac{1-a}{2} + 2b \right\} ( 2T(\Psi) +1 ).
    \end{equation}
    This establishes the equivalence between the square root of $\langle \Psi,H'\Psi \rangle$ and the $H^1$-norm of $\Psi$.
\end{proof}

\begin{theorem}\label{th:minimizer}
    Let $v\in\affspace^*$ then the corresponding potential form $V$ defines a self-adjoint Hamiltonian $H=H_0+V$ that is bounded below. This $H$ has a minimizer $\Psi_0\in\psiset$ in Eq.~\eqref{eq:E-def}, $E(v) = E_0 = \langle \Psi_0, (H_0+V) \Psi_0 \rangle$, that is also a solution to the Schrödinger equation $H\Psi_0=E_0\Psi_0$ (in a distributional sense).
\end{theorem}

\begin{proof}
    By Lemma~\ref{lem:norm-equiv} we know that any $v\in\affspace^*$ defines a self-adjoint Hamiltonian that is bounded below.
    Next, take a sequence $\{\Psi_n\}_n \in \psiset$ that realizes the infimum in Eq.~\eqref{eq:E-def}, so $\langle \Psi_n, H\Psi_n \rangle \to E_0$. Consequently, for any $\eps>0$ there is a $n_0\in\N$ such that for all $n>n_0$ it holds $\langle \Psi_n, H\Psi_n \rangle < E_0+\eps$. 
    By Lemma~\ref{lem:v-kinetic-bound} for $V$ and by assumption for $W$, both quadratic forms have relative kinetic bounds $<1$, so we can find $a_V,a_W <1$ like in Definition~\ref{def:kinetic-bound} where $a_V$ can even be taken arbitrarily small, so $a_V+a_W<1$. With the estimate from the kinetic bound it holds for all $n>n_0$ that
    \begin{equation}
        E_0+\eps > \langle \Psi_n,H\Psi_n \rangle = T(\Psi_n) + \langle \Psi_n,W\Psi_n \rangle + \langle \Psi_n,V\Psi_n \rangle \geq (1-a_W-a_V)T(\Psi_n) - b_W-b_V.
    \end{equation}
    But for all (normalized) $\Psi_n\in\psiset$ the $T(\Psi_n)$ is equivalent to the $H^1$-norm squared, so the above inequality means that all $\Psi_n$ with $n>n_0$ are in a $H^1$-bounded set.
    Since the Sobolev space $H^1$ is reflexive, this guarantees that $\{\Psi_n\}_n$ has a subsequence that converges weakly in $H^1$ \cite[Cor.~11.9]{book-clason}.
    Further, since the torus is a bounded domain, by the Rellich--Kondrachov theorem \cite[Th.~6.3]{adams-book} $H^1$ is compactly embedded in $L^2$. This means the previously chosen subsequence permits another subsequence, again called $\{\Psi_n\}_n$, that converges strongly in $L^2$ to a limit that we call $\Psi_0$. Because of the weak convergence in $H^1$ and the strong convergence in $L^2$ we know that $\Psi_0\in\psiset$ which implies $\|\Psi_0\|_{L^2}=1$.
    From Lemma~\ref{lem:norm-equiv} we now know that with an appropriate constant shift we can define $H'=H+c$, where the square root of the expectation value $\langle \Psi,H'\Psi \rangle$ defines a norm equivalent to the $H^1$-norm. Since the norm is non-increasing under weak limits \cite[Thm.~$11.2$]{book-clason} and any expectation value of $H$ with respect to $\Psi_0\in\psiset$ must be bounded below by $E_0$, we have
    \begin{equation}
        E_0+c \leq \langle \Psi_0,(H+c)\Psi_0 \rangle = \langle \Psi_0,H'\Psi_0 \rangle \leq \lim_{n\to\infty} \langle \Psi_n,H'\Psi_n \rangle = E_0+c
    \end{equation}
    and thus $E_0=\langle \Psi_0,H\Psi_0 \rangle$.\\
    For showing that this $\Psi_0$ solves the Schrödinger equation in a distributional sense we follow \citet[Th.~11.8]{Lieb-Loss-book}. Take any test function $\Phi$ and note with $\alpha\in\R$ from the variational principle of Eq.~\eqref{eq:E-def} that
    \begin{equation}
        E_0 \|\Psi_0+\alpha\Phi\|^2_{\H} \leq \langle (\Psi_0+\alpha\Phi),H(\Psi_0+\alpha\Phi) \rangle,
    \end{equation}
    where the norm of the vector $\Psi_0+\alpha\Phi$ on the left appears due to normalization.
    Writing the squares all out and using that $E_0=\langle\Psi_0, H\Psi_0\rangle$ we have
    \begin{equation}
        0 \leq 2\alpha\Re \langle\Phi,(H-E_0)\Psi_0\rangle + \alpha^2 (\langle\Phi,H\Phi\rangle - E_0\|\Phi\|^2_\H).
    \end{equation}
    Now this equation needs to hold for $\alpha$ arbitrarily small and so it follows
    \begin{equation}
        0 \leq 2\alpha\Re \langle\Phi,(H-E_0)\Psi_0\rangle
    \end{equation}
    and since $\alpha$ can have any sign we conclude that $\Re \langle\Phi,(H-E_0)\Psi_0\rangle = 0$. Replacing $\Phi$ by $\i\Phi$ gives $\Im \langle\Phi,(H-E_0)\Psi_0\rangle=0$ and consequently $\langle\Phi,(H-E_0)\Psi_0\rangle=0$ for arbitrary test functions $\Phi$ which concludes the proof.
\end{proof}

\section{Conclusions and outlook}
\label{sec:conclusions}

To prove our main result from the introduction is now just a matter of collecting the results from the previous sections.

\begin{proof}[Proof of Theorem~\ref{th:main}]
Let $\rho \in \vrepset$ then Theorem~\ref{th:F-nonempty-subdiff} shows that $\rho\in\densset$ and $\partial F(\rho)$ is non-empty. Take $-v\in \partial F(\rho)\subset \affspace^*$ then $\rho$ is a minimizer in the variational problem Eq.~\eqref{eq:sec:rep:E-from-F} and $E(v)=F(\rho) + \langle v,\rho \rangle$ (see Section~\ref{sec:rep}). In order to connect to quantum-mechanical ground states we need the constrained-search functional. We cannot use $\tilde F$, since $F\neq \tilde F$ in general, but we have $F=\FDM$ on $\densset$ by Theorem~\ref{th:FDM}. By this result,
\begin{equation}\label{eq:th:main:proof:E0}
    E_0 = E(v) = \FDM(\rho) + \langle v,\rho \rangle = \inf_{\Gamma\mapsto\rho}\{\Tr \Gamma H_0\} + \langle v,\rho \rangle = \inf_{\Gamma\mapsto\rho}\{\Tr \Gamma H\},
\end{equation}
where we simply recombined $H_0+V=H$.
Now as in Theorem~\ref{th:minimizer} we next want to show that this infimum allows for a minimizer $\Gamma_0\mapsto\rho$. This is achieved by \citet[Cor.~4.5(ii)]{Lieb1983} that can be applied in the current setting since $\FDM$ is defined in exactly the same manner (with the only difference of a one-dimensional torus domain instead of $\R^3$ which only simplifies things) and every $\rho \in \vrepset$ is also in $\densset$ by Theorem~\ref{th:F-nonempty-subdiff}, the set considered in the reference.
Now, by Theorem~\ref{th:FDM} such a $\Gamma_0\mapsto\rho$ can be constructed as a (possibly infinite) convex sum $\Gamma_0=\sum \lambda_k \langle \Psi_{0,k},\cdot \rangle \Psi_{0,k}$ with $\Psi_{0,k} \in \psiset$ all orthogonal. Assume that there is one $\Psi_{0,l}$ in this set with $\lambda_l>0$ that has an energy expectation value that is above the ground-state energy, i.e., $\langle \Psi_{0,l},H\Psi_{0,l} \rangle = E_0+\eps$ with $\eps>0$. Then it follows immediately from Eq.~\eqref{eq:th:main:proof:E0} that $E_0 = \sum_k\lambda_k E_0 + \lambda_l\eps = E_0+\lambda_l\eps$, a clear contradiction. Consequently, all $\Psi_{0,k}$ that appear in $\Gamma_0$ with $\lambda_k>0$ have $\langle \Psi_{0,k},H\Psi_{0,k} \rangle = E_0$ and thus are ground-state solutions for the Schrödinger equation $H\Psi_{0,k}=E_0\Psi_{0,k}$ in a distributional sense as demonstrated in Theorem~\ref{th:minimizer}.
Since it also holds from $\Gamma_0\mapsto\rho$ that $\sum\lambda_k\rho_k=\rho$ with $\Psi_{0,k}\mapsto \rho_k$, $\rho$ is indeed $v$-representable. In case we need more than a single $\Psi_{0,k}$ to arrive at $\rho$, $v$ yields a degenerate ground state and $\rho$ is the density of an ensemble state.
\end{proof}

That this result only holds for one-dimensional periodic domains can of course be seen as a severe limitation, yet the considered torus geometries are customarily employed in condensed matter physics in the form of Born--von-Kármán periodic boundary conditions. This means the presented theory of $v$-representability, apart from solving this important problem of DFT, is also immediately relevant for models of condensed matter physics.
Still, the appearance of a distributional potential might seem unconventional, but one has to remember that already one of the most early models of condensed matter physics, the Kronig--Penney model~\cite{KP1931}, included a limit to \emph{distributional} potentials in the form of delta peaks (cf.~Example~\ref{ex:delta}). Similar `point interactions' are also considered in the theory of one-dimensional Schrödinger operators~\cite{Kostenko2010} and are among the very few analytically solvable models of quantum mechanics~\cite{albeverio2012solvable}. Another famous appearance of distributional potentials is the Haldane pseudopotential~\cite{Haldane1983,Seiringer2020} that manifests the Laughlin wavefunction of the fractional quantum Hall effect as an exact eigenstate. This list can be continued with the Lieb--Liniger model~\cite{Lieb1963} for one-dimensional bosons that leads to the Gross--Pitaevskii equation in an effective description~\cite{Kopyciski2022}. All those examples demonstrate that distributional potentials are an indispensable tool to describe important quantum phenomena. Returning to the setting of DFT, this translates to the realization that in order to receive perfectly well-behaved densities (like the one from the example at the outset of the paper, Eq.~\eqref{eq:EE-dens}) as the result of a ground-state problem, one has to include distributional potentials into the quantum-mechanical description.

With our main result proven and brought into context, we want to conclude with a list of follow-up questions and ideas how to generalize our findings.

\begin{enumerate}
    \item A first open question would be that of the `maximality' of the proposed $v$-representable density set $\vrepset$. So, is $\rho\in\affspace$, $\rho\geq 0$, with $\rho(x)=0$ at one or many $x\in\T$ not $v$-representable by any $v\in\affspace^*$? This is equivalent to the statement that any $v\in\affspace^*$ yields a ground-state density that is gapped away from zero, i.e., $\rho\in\vrepset$. Such a result would close the circle and allow for a well-defined mapping between the potential and density sets $\affspace^*$ and $\vrepset$ (yet not necessarily a one-to-one mapping due to potential non-unique $v$-representability and the possible occurrence of degeneracies).

    \item Another obvious open question is the possibility of a generalization to a $(d\geq 2)$-dimensional torus domain or other more-dimensional compact manifolds. The current method critically uses the embedding  $H^{1} = W^{1,2} \hookrightarrow L^\infty$ from Lemma~\ref{lem:embeddings} which only holds for $d=1$. This is then employed for showing that $\vrepset$ is open (Lemma~\ref{lem:open}) and also to have a kinetic bound for potentials from $\affspace^*$ (Lemma~\ref{lem:v-kinetic-bound}).
    Other Sobolev embeddings are available, and in $d\geq 2$ dimensions $W^{1,p}\hookrightarrow L^\infty$ for $p> d$ could be applied.
    But Lemma~\ref{lem:v-kinetic-bound} also critically relies on an estimate of the $W^{1,2}$-norm by the kinetic energy, which is again needed to show applicability of the KLMN theorem, and the same is not achievable for $W^{1,p}$ with $p>2$. Consequently, a generalization to higher dimensions requires a different approach. In any case, we would not expect that the basic finding of this work would change in higher dimensions: That in order to represent a reasonable class of densities with finite kinetic energy, distributional potentials have to be considered.

    \item A crucial step in the proof of \citet{hohenberg-kohn1964} is to show that for two potentials that differ by more than a constant, the ground states cannot be possibly equal. In \citet{Penz-et-al-HKreview-partI} this was termed ``HK2'' theorem and formulated equivalently: If two potentials share any common eigenstate then the potentials are equal up to a constant. But \citet{hohenberg-kohn1964} do not give any details about this part of the proof, which usually proceeds as follows. Assume $\Psi$ is an eigenstate of two different potentials $V=\sum_j v(x_j)$ and $V'=\sum_j v'(x_j)$ then it must hold
    \begin{align}
        &(H_0+V)\Psi=E\Psi \\
        &(H_0+V')\Psi=E'\Psi,
    \end{align}
    and subtraction of the two equations gives $(V-V')\Psi = (E-E')\Psi$. Then, moving all potential parts that do not depend on $x_1$ to one side gives
	\begin{equation}
		\left(v(x_1)-v'(x_1)\right)\Psi = (E-E')\Psi - \sum_{j=2}^N \left(v(x_j)-v'(x_j)\right)\Psi
	\end{equation}
	and one could infer that $v(x_1)-v'(x_1)=\mathrm{const}$ after division by $\Psi$. This division is the crucial step here, and we assume that it was this part of the proof that \citet{CCR1985} had in in mind when they wrote about including generalized potentials, as already partly cited in the beginning: ``However, we caution the reader that much of HK theory, particularly the proof of the basic uniqueness theorem, may break down for generalized potentials.''
    Now the division step firstly requires the wave function to be non-zero almost everywhere, which is secured by the unique continuation property in certain settings~\cite{Garrigue2018HK}, but also it only works if the potential is a multiplication operator. (\citet{Dreizler-Gross}, who repeat the example of \citet{ENGLISCH1983}, after Eq.~(2.24) comment on precisely this problem, but they put the focus on a possible failure of the unique continuation property instead.) As we showed here, in order to have $v$-representability, we must extend the domain of potentials to distributions that are not multiplicative, hence this step in the proof of the Hohenberg--Kohn theorem will generally fail. While this is not a falsification of Hohenberg--Kohn it is still an important indication that something might go wrong here. Similarly, in finite-lattice systems it is \emph{known} that there are counterexamples to the Hohenberg--Kohn theorem while one has full $v$-representability for all densities $0<\rho<1$ \cite{penz-DFT-graphs}. So while we gain $v$-representability we actually might loose the unique mapping from densities to potentials. In this context it is interesting to remember the observation of Simen Kvaal in \citet[2.1.17]{teale2022dft} that neither the convex formulation of exact DFT nor the Kohn--Sham approach do in any way rely on the Hohenberg--Kohn theorem. Rather they depend on the constraint-search functional and the Legendre--Fenchel transformation and, in the case of Kohn--Sham, specifically on interacting and non-interacting $v$-representability. This means that the given result is imperative for a well-defined Kohn--Sham procedure, since it proves the existence of the $v_\mathrm{s}$ potential of the auxiliary system for the first time (cf.~Mathieu Lewin in \citet[4.5.2]{teale2022dft}). Another application is the adiabatic connection where indeed $v$-representability for different interaction strengths needs to be assumed, a property that can be achieved in the setting given here. In this sense it can be argued that the achievement of $v$-representability outweighs the potential loss of the validity of the Hohenberg--Kohn theorem.

    \item The previous question on the validity of the Hohenberg--Kohn theorem for the class $\affspace^*$ of potentials even connects directly to functional differentiability by another result of convex analysis. Since $F(\rho)$ is continuous at any $\rho\in\vrepset$ we have a non-empty subdifferential $\partial F(\rho)$ (Theorem~\ref{th:F-nonempty-subdiff}) and as a consequence $\rho$ is $v$-representable. If $\partial F(\rho)$ contains just a \emph{single} element, which would be guaranteed by the Hohenberg--Kohn theorem, then this implies that $F$ is also Gateaux differentiable at $\rho$, and vice versa \cite[Prop.~2.40]{Barbu-Precupanu}. Most formulations of DFT rely implicitly on this functional differentiability. The $v$-representability is a necessary but not a sufficient criterion for functional differentiability and the previous point even gives some indications for a possible failure of the Hohenberg--Kohn theorem, so the question about functional differentiability still remains unanswered. Nevertheless we note the following interesting corollary.
    \begin{corollary}
        The two statements are equivalent:
        \begin{enumerate}[(i)]
            \item Every $\rho\in\vrepset$ is represented by a \emph{unique} potential $v\in\affspace^*$ (Hohenberg--Kohn theorem).
            \item $F$ is Gateaux differentiable at all $\rho\in\vrepset$.
        \end{enumerate}
    \end{corollary}

    \item Example~\ref{ex:potentials-act} in Section~\ref{sec:dens} already showed how potentials from $\affspace^*$ act on wave functions and densities. But it would be interesting to find ways how to implement such potentials numerically. Since the example showed that the potential has a semi-local effect, this brings the $v$-representing potential close to the realm of gradient expansions~\cite{LangrethMehl1983, SvendsenBarth1996} and generalized gradient approximations~\cite{Perdew1985GGA, Perdew1986GGA, Becke1988GGA}.
    In practice, such functionals often get mixed with non-multiplicative exchange operators to form `hybrid' functionals that, in different variants, belong to the most widely employed density functionals in chemistry and condensed matter physics~\cite{Csonka2010}. While originally being devised as a mere heuristics, they later received a rigorous treatment in the form of \emph{generalized} Kohn--Sham theory~\cite{Seidl1996,Garrick2020}. Just like the usual Kohn--Sham scheme, this method is also in principle exact and aims at reproducing the ground-state density of the fully interacting system, just with a manifestly non-multiplicative exchange operator and a semilocal approximation to the exchange-correlation potential. In the light of our results, this is not just a technicality, but a mathematical necessity of the theory.

    \item The main difference to other formulations of DFT is the switch from the density space $L^1\cap L^3$ as in \citet{Lieb1983} to one with a finer topology, the Sobolev space $H^1$. This means the meaning of `closeness' changes, and densities that have been nearby in the former topology, and led to non-differentiability, are now remote, thus allowing for possible differentiability of the functional. The reason is that the new topology measures the energy content of the associated wave function more accurately, as expressed through Lemma~\ref{lem:F-bound} and the embeddings of Lemma~\ref{lem:embeddings}. This can have important consequences for other problems in DFT such as the question of convergence of the Kohn--Sham iteration scheme \cite{penz2019guaranteed,penz2020erratum,Lammert-bivariate} where, after all, it is always convergence with respect to a certain topology.
    Further and interestingly, these Sobolev spaces appeared similarly in a recent reformulation of the Zhao--Morrison--Parr method of density-potential inversion in terms of Moreau--Yosida regularization~\cite{penz2023MY-ZMP}, although exactly in reverse. There, the density space is $H^{-1}$ and the potentials are taken from $H^1$, as it would be according to the Poisson equation $-\Delta v = 4\pi\rho$.

    \item If one considers a Kohn--Sham system then the potential from a density-potential inversion is usually decomposed as $v=\vext + \vH + \vxc$, the external potential, the Hartree potential and the exchange-correlation contribution. The Hartree potential fulfils (now in a 3d continuous setting) the Poisson equation $-\Delta\vH = 4\pi\rho$ and so for $\rho\in H^1$ it holds $\vH\in H^3$. Analogous to that the exchange-correlation potential is sometimes considered as the effect of a fictitious `xc density', $-\Delta\vxc = 4\pi\rhoxc$ \cite{Andrade2011}. But assuming a regular $\vext$ we thus have from a general $v\in H^{-1}$ that also $\vxc\in H^{-1}$ which means $\rhoxc\in H^{-3}$ which is anything but a usual density distribution. So the current analysis indicates that the exchange-correlation potential cannot be generally assumed to be the effect of a fictitious charge density. 

    \item The formulation can be extended to general particle numbers by extending the constrained search over the full Fock space. Densities integrating to a non-integer number would not be $N$-representable anymore in a strict sense, since there is not an integer-particle state generating that density. Yet, this definition is also readily extended to Fock space.
    A notable difference compared to the presented formulation is that the constant of the potential determines the particle number to some extent: to each integer density a closed interval of constants would be associated and for non-integer densities a unique constant. So the potential space would not be a simple equivalence class of $H^{-1}$, but have a more involved equivalence class structure. This new complication could be avoided by going all the way to a grand canonical setting with a finite temperature. Since the entropy can by bounded by the kinetic energy \citep[Appendix E]{seidl2017sphericalSCE}, one could still work with $\rho \in \densset$. The entropy works as a smoothener, similar to the Moreau--Yosida regularization \citep{penz2023MY-ZMP}, so might even lead to some form of differentiability as in its one-body reduced density matrix counterpart \citep{GiesbertzRuggenthaler2019, SutterGiesbertz2023}.
\end{enumerate}

\begin{acknowledgments}
The initial ideas for this paper were born at a joint research visit of M.\ Penz, M.\ Ruggenthaler and R.\ van Leeuwen in November 2021 at the Oberwolfach Research Institute for Mathematics as a part of their ``Research in Pairs'' program.
S.~M.\ Sutter and K.~J.~H.\ Giesbertz acknowledge support by the Netherlands Organisation for Scientific Research (NWO) under Vici grant 724.017.001 and KLEIN-1 grant OCENW.KLEIN.434. R.\ van Leeuwen acknowledges the Academy of Finland grant under project number 356906. M.\ Ruggenthaler acknowledges the Cluster of Excellence ``CUI: Advanced Imaging of Matter'' of the Deutsche Forschungsgemeinschaft (DFG), EXC 2056, project ID 390715994.
\end{acknowledgments}

\begin{appendix}
\section{Kinetic bounds for interactions}
\label{app:kinetic-bounds-interactions}

\begin{lemma}
The two-body multiplication operator $W=\sum _{i<j}w(x_i-x_j)$ with $w \in L^1([-1,1])$ is kinetically bounded with relative kinetic bound $0$.
\end{lemma}

\begin{proof}
Without loss of generality we only consider the spinless case and limit ourselves to the $w(x_1-x_2)$ term. The idea is to rotate the coordinates such that the two-body interaction becomes a one-body operator and then proceed as for the non-distributional potentials. This is achieved with the substitution $y_1 = (x_1-x_2)/\sqrt{2}$ and $y_2=(x_1+x_2)/\sqrt{2}$ which corresponds to the change of coordinates $(x_1,\ldots,x_N) = r(y_1,\ldots , y_N) = ((y_1+y_2)/\sqrt{2}, (y_2-y_1)/\sqrt{2}, y_3, \ldots, y_N)$. The integration domain $\T^N \simeq [0,1]^N$ for $(x_i)_i$ then changes to $[-1/\sqrt{2},1/\sqrt{2}] \times [0,1/\sqrt{2}] \times [0,1]^{N-2}$ for $(y_i)_i$ (notation for this will be suppressed in the integrals), which is displayed in Fig.~\ref{fig:app:int-domain}.
\begin{figure}[ht]
	\centering
	\resizebox{.4\columnwidth}{!}{%
	\begin{tikzpicture}[every node/.style={scale=1.5}] 
	    \draw[->] (-3.5,0) -- (6,0) node[below] {$x_1$};
        \draw[->] (0,-3.5) -- (0,6) node[left] {$x_2$};
        \draw[->] (-3,-3) -- (5.5,5.5) node[right] {$y_2$};
        \draw[->] (-3.5,3.5) -- (3.5,-3.5) node[right] {$y_1$};
        \draw (0.2,5) -- (-0.2,5) node[left] {$1$};
        \draw (5,0.2) -- (5,-0.2) node[below] {$1$};
        \draw[pattern=north west lines, pattern color=blue, opacity=0.0, fill opacity=0.2] (2.5,2.5) -- (5,5) -- (5,0) -- (2.5,2.5);
        \draw[pattern=north west lines, pattern color=blue, opacity=0.0, fill opacity=0.2] (-2.5,2.5) -- (0,5) -- (0,0) -- (-2.5,2.5);
        \draw[pattern=north west lines, pattern color=red, opacity=0.0, fill opacity=0.2] (2.5,2.5) -- (5,5) -- (0,5) -- (2.5,2.5);
        \draw[pattern=north west lines, pattern color=red, opacity=0.0, fill opacity=0.2] (0,0) -- (5,0) -- (2.5,-2.5) -- (0,0);
        \draw[-,line width=2pt] (0,0) -- (0,5) -- (5,5) -- (5,0) -- (0,0);
        \draw[-,color=blue,line width=2pt] (-2.5,2.5) -- (0,5) -- (5,0) -- (2.5,-2.5) -- (-2.5,2.5);
	\end{tikzpicture}
	}
	\caption{Change of the periodic integration domain when the $(x_i)_i$ coordinates (black square) are changed to $(y_i)_i$ (tilted blue rectangle). The two domains are equivalent, since the two hatched blue (red) areas are equivalent due to periodicity of the torus domain $\T^2 \simeq [0,1]^2$.}
	\label{fig:app:int-domain}
\end{figure}
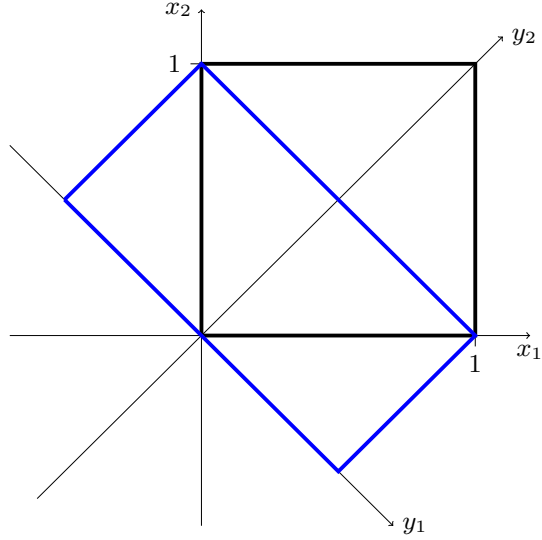
We then have
\begin{equation}
\begin{aligned}
\label{eq:app:w-expectation-rotation}
    \innerproduct{\Psi, w(x_1-x_2)\Psi} &= \int \mathrm{d}x_1 \ldots \mathrm{d}x_N \abs{\Psi}^2w(x_1-x_2) \\
    &=\int \mathrm{d}y_1 \ldots \mathrm{d}y_N \abs{(\Psi\circ r)(y_1, \ldots, y_N)}^2w(\sqrt{2}y_1) = \langle w(\sqrt{2}\cdot),\rho_r \rangle,
\end{aligned}
\end{equation}
where we introduced $\rho_r$ as the $y_1$-marginal of $|\Psi\circ r|^2$. Since $w \in L^1([-1,1])$ there is a sequence $w_n \in C^\infty([-1,1])$ such that $\norm{w-w_n}_{L^1} \to 0$. By inserting $w=w_n + (w - w_n)$, Eq.~\eqref{eq:app:w-expectation-rotation} can be further estimated using the Sobolev embedding $H^1(\T) \hookrightarrow L^\infty(\T)$ from Lemma~\ref{lem:embeddings} for $\|\rho_r\|_{L^\infty} = \|\sqrt{\rho_r}\|_{L^\infty}^2$ and the estimate in terms of the wave function from \citet[Th.~1.1]{Lieb1983}.
\begin{equation}\label{eq:app:w-rhor-estimate}
\begin{aligned}
    |\langle w(\sqrt{2}\cdot),\rho_r \rangle|
    &\leq \norm{w_n}_{L^\infty} \|\rho_r\|_{L^1} + \frac{1}{\sqrt{2}} \norm{w-w_n}_{L^1} \|\rho_r\|_{L^\infty} \\
    &\leq \norm{w_n}_{L^\infty} \|\rho_r\|_{L^1} + C \norm{w-w_n}_{L^1} \left( \|\sqrt{\rho_r}\|_{L^2}^2 + \|\nabla\sqrt{\rho_r}\|_{L^2}^2 \right)\\[0.5em]
    &\leq N\left( \norm {w_n}_{L^\infty} + C \norm{w-w_n}_{L^1} \right) \norm{\Psi\circ r}_{L^2}^2+C N \norm{w-w_n}_{L^1} \norm{\nabla_{y_1} (\Psi \circ r)}_{L^2}^2\\
    &\leq N\left(\norm{w_n}_{L^\infty}+C\norm{w-w_n}_{L^1}\right) \norm{\Psi}_{L^2}^2 +\frac{C N}{2} \norm{w-w_n}_{L^1}(\norm{\nabla_{x_1} \Psi}_{L^2}^2 + \norm{\nabla_{x_2} \Psi}_{L^2}^2)
\end{aligned}
\end{equation}
In the last line, we just used $\nabla_{y_1} = (\nabla_{x_1}-\nabla_{x_2})/\sqrt{2}$ and the triangle inequality.
We have $\norm{\nabla_{x_j} \Psi}_{L^2}^2 \leq 2T(\Psi)$ and $n$ can be chosen large enough such that $\norm{w-w_n}_{L^1}$ becomes arbitrarily small. Thus, we conclude the zero kinetic bound.
\end{proof}

This result covers all typical interactions that depend on the distance between two particles. The next one even extends this to distributional potentials of the form $\affspace^*$.

\begin{lemma}
A distributional interaction of the form $W=\sum _{i<j} \nabla_{x_i} g(x_i-x_j)$ with $g \in L^2([-1,1])$ is kinetically bounded with relative kinetic bound $0$.
\end{lemma}

\begin{proof}
As in the previous proof we only need to consider the spinless case and use the same coordinate transformation $r:(y_1,\ldots,y_N)\mapsto (x_1,\ldots,x_N)$ to get
\begin{equation}
    \nabla_{x_1} g(x_1-x_2) = \frac{1}{2}\left( \nabla_{x_1} g(x_1-x_2) - \nabla_{x_2} g(x_1-x_2)\right) = \frac{1}{\sqrt{2}} \nabla_{y_1} g(\sqrt{2}y_1) =: \nabla_{y_1}\tilde g(y_1).
\end{equation}
We can thus write $\innerproduct{\Psi, (\nabla_{x_1} g(x_1-x_2)) \Psi} = \innerproduct{\nabla \tilde g, \rho_r}$ and then continue exactly like in the proof of Lemma~\ref{lem:v-kinetic-bound}. By introducing a sequence $\{\tilde g_n\}_n$ in $C^\infty([-1,1])$ that has $\|\tilde g-\tilde g_n\|_{L^2} \to 0$ we get an estimate 
\begin{equation}
     |\innerproduct{\nabla \tilde g, \rho_r}| \leq  \|\nabla \tilde g_n\|_{L^\infty} \|\rho_r\|_{L^1} + 2C \|\tilde g-\tilde g_n\|_{L^2} \left(\|\rho_r\|_{L^1} + \|\nabla\sqrt{\rho_r}\|_{L^2}^2 \right).
\end{equation}
From here on, the same estimates as in Eq.~\eqref{eq:app:w-rhor-estimate} apply which concludes the kinetic boundedness.
\end{proof}

A well-known example for such a distributional interaction is the $\delta$-interaction. On the torus the distribution $\delta(x)+1$ is given by $\nabla g(x)$ where $g(x)=1-x$ which is then periodically extended to $[-1,1]$ for a proper domain of the interaction. Such a system is studied in \citet{gaudin1967systeme}.

\end{appendix}


\begin{thebibliography}{57}%
\makeatletter
\providecommand \@ifxundefined [1]{%
 \@ifx{#1\undefined}
}%
\providecommand \@ifnum [1]{%
 \ifnum #1\expandafter \@firstoftwo
 \else \expandafter \@secondoftwo
 \fi
}%
\providecommand \@ifx [1]{%
 \ifx #1\expandafter \@firstoftwo
 \else \expandafter \@secondoftwo
 \fi
}%
\providecommand \natexlab [1]{#1}%
\providecommand \enquote  [1]{``#1''}%
\providecommand \bibnamefont  [1]{#1}%
\providecommand \bibfnamefont [1]{#1}%
\providecommand \citenamefont [1]{#1}%
\providecommand \href@noop [0]{\@secondoftwo}%
\providecommand \href [0]{\begingroup \@sanitize@url \@href}%
\providecommand \@href[1]{\@@startlink{#1}\@@href}%
\providecommand \@@href[1]{\endgroup#1\@@endlink}%
\providecommand \@sanitize@url [0]{\catcode `\\12\catcode `\$12\catcode
  `\&12\catcode `\#12\catcode `\^12\catcode `\_12\catcode `\%12\relax}%
\providecommand \@@startlink[1]{}%
\providecommand \@@endlink[0]{}%
\providecommand \url  [0]{\begingroup\@sanitize@url \@url }%
\providecommand \@url [1]{\endgroup\@href {#1}{\urlprefix }}%
\providecommand \urlprefix  [0]{URL }%
\providecommand \Eprint [0]{\href }%
\providecommand \doibase [0]{http://dx.doi.org/}%
\providecommand \selectlanguage [0]{\@gobble}%
\providecommand \bibinfo  [0]{\@secondoftwo}%
\providecommand \bibfield  [0]{\@secondoftwo}%
\providecommand \translation [1]{[#1]}%
\providecommand \BibitemOpen [0]{}%
\providecommand \bibitemStop [0]{}%
\providecommand \bibitemNoStop [0]{.\EOS\space}%
\providecommand \EOS [0]{\spacefactor3000\relax}%
\providecommand \BibitemShut  [1]{\csname bibitem#1\endcsname}%
\let\auto@bib@innerbib\@empty
\bibitem [{\citenamefont {Hohenberg}\ and\ \citenamefont
  {Kohn}(1964)}]{hohenberg-kohn1964}%
  \BibitemOpen
  \bibfield  {author} {\bibinfo {author} {\bibfnamefont {P.}~\bibnamefont
  {Hohenberg}}\ and\ \bibinfo {author} {\bibfnamefont {W.}~\bibnamefont
  {Kohn}},\ }\bibfield  {title} {\emph {\bibinfo {title} {Inhomogeneous
  electron gas},\ }}\href {\doibase 10.1103/PhysRev.136.B864} {\bibfield
  {journal} {\bibinfo  {journal} {Phys. Rev.}\ }\textbf {\bibinfo {volume}
  {136}},\ \bibinfo {pages} {B864} (\bibinfo {year} {1964})}\BibitemShut
  {NoStop}%
\bibitem [{\citenamefont {Levy}(1979)}]{levy1979universal}%
  \BibitemOpen
  \bibfield  {author} {\bibinfo {author} {\bibfnamefont {M.}~\bibnamefont
  {Levy}},\ }\bibfield  {title} {\emph {\bibinfo {title} {Universal variational
  functionals of electron densities, first-order density matrices, and natural
  spin-orbitals and solution of the v-representability problem},\ }}\href
  {\doibase 10.1073/pnas.76.12.6062} {\bibfield  {journal} {\bibinfo  {journal}
  {Proc. Natl. Acad. Sci.}\ }\textbf {\bibinfo {volume} {76}},\ \bibinfo
  {pages} {6062} (\bibinfo {year} {1979})}\BibitemShut {NoStop}%
\bibitem [{\citenamefont {Levy}(1982)}]{levy1982electron}%
  \BibitemOpen
  \bibfield  {author} {\bibinfo {author} {\bibfnamefont {M.}~\bibnamefont
  {Levy}},\ }\bibfield  {title} {\emph {\bibinfo {title} {Electron densities in
  search of {H}amiltonians},\ }}\href {\doibase 10.1103/PhysRevA.26.1200}
  {\bibfield  {journal} {\bibinfo  {journal} {Phys. Rev. A}\ }\textbf {\bibinfo
  {volume} {26}},\ \bibinfo {pages} {1200} (\bibinfo {year}
  {1982})}\BibitemShut {NoStop}%
\bibitem [{\citenamefont {Lieb}(1983)}]{Lieb1983}%
  \BibitemOpen
  \bibfield  {author} {\bibinfo {author} {\bibfnamefont {E.~H.}\ \bibnamefont
  {Lieb}},\ }\bibfield  {title} {\emph {\bibinfo {title} {Density functionals
  for {C}oulomb-systems},\ }}\href {\doibase 10.1002/qua.560240302} {\bibfield
  {journal} {\bibinfo  {journal} {Int. J. Quantum Chem.}\ }\textbf {\bibinfo
  {volume} {24}},\ \bibinfo {pages} {243} (\bibinfo {year} {1983})}\BibitemShut
  {NoStop}%
\bibitem [{\citenamefont {Englisch}\ and\ \citenamefont
  {Englisch}(1983)}]{ENGLISCH1983}%
  \BibitemOpen
  \bibfield  {author} {\bibinfo {author} {\bibfnamefont {H.}~\bibnamefont
  {Englisch}}\ and\ \bibinfo {author} {\bibfnamefont {R.}~\bibnamefont
  {Englisch}},\ }\bibfield  {title} {\emph {\bibinfo {title}
  {{H}ohenberg--{K}ohn theorem and non-{V}-representable densities},\ }}\href
  {\doibase 10.1016/0378-4371(83)90254-6} {\bibfield  {journal} {\bibinfo
  {journal} {Physica A Stat. Mech. Appl.}\ }\textbf {\bibinfo {volume} {121}},\
  \bibinfo {pages} {253 } (\bibinfo {year} {1983})}\BibitemShut {NoStop}%
\bibitem [{\citenamefont {Kohn}\ and\ \citenamefont {Sham}(1965)}]{KS1965}%
  \BibitemOpen
  \bibfield  {author} {\bibinfo {author} {\bibfnamefont {W.}~\bibnamefont
  {Kohn}}\ and\ \bibinfo {author} {\bibfnamefont {L.~J.}\ \bibnamefont
  {Sham}},\ }\bibfield  {title} {\emph {\bibinfo {title} {Self-consistent
  equations including exchange and correlation effects},\ }}\href {\doibase
  10.1103/PhysRev.140.A1133} {\bibfield  {journal} {\bibinfo  {journal} {Phys.
  Rev.}\ }\textbf {\bibinfo {volume} {140}},\ \bibinfo {pages} {A1133}
  (\bibinfo {year} {1965})}\BibitemShut {NoStop}%
\bibitem [{\citenamefont {Wrighton}\ \emph {et~al.}(2023)\citenamefont
  {Wrighton}, \citenamefont {Albavera-Mata}, \citenamefont {Rodr{\'\i}guez},
  \citenamefont {Tan}, \citenamefont {Cancio}, \citenamefont {Dufty},\ and\
  \citenamefont {Trickey}}]{wrighton2023some}%
  \BibitemOpen
  \bibfield  {author} {\bibinfo {author} {\bibfnamefont {J.}~\bibnamefont
  {Wrighton}}, \bibinfo {author} {\bibfnamefont {A.}~\bibnamefont
  {Albavera-Mata}}, \bibinfo {author} {\bibfnamefont {H.~F.}\ \bibnamefont
  {Rodr{\'\i}guez}}, \bibinfo {author} {\bibfnamefont {T.~S.}\ \bibnamefont
  {Tan}}, \bibinfo {author} {\bibfnamefont {A.~C.}\ \bibnamefont {Cancio}},
  \bibinfo {author} {\bibfnamefont {J.}~\bibnamefont {Dufty}}, \ and\ \bibinfo
  {author} {\bibfnamefont {S.}~\bibnamefont {Trickey}},\ }\bibfield  {title}
  {\emph {\bibinfo {title} {Some problems in density functional theory},\
  }}\href {\doibase 10.1007/s11005-023-01649-z} {\bibfield  {journal} {\bibinfo
   {journal} {Lett. Math. Phys.}\ }\textbf {\bibinfo {volume} {113}},\ \bibinfo
  {pages} {41} (\bibinfo {year} {2023})}\BibitemShut {NoStop}%
\bibitem [{\citenamefont {Lewin}\ \emph {et~al.}(2023)\citenamefont {Lewin},
  \citenamefont {Lieb},\ and\ \citenamefont
  {Seiringer}}]{lewin-2023-in-DFT-book}%
  \BibitemOpen
  \bibfield  {author} {\bibinfo {author} {\bibfnamefont {M.}~\bibnamefont
  {Lewin}}, \bibinfo {author} {\bibfnamefont {E.~H.}\ \bibnamefont {Lieb}}, \
  and\ \bibinfo {author} {\bibfnamefont {R.}~\bibnamefont {Seiringer}},\ }in\
  \href {\doibase 10.1007/978-3-031-22340-2_3} {\emph {\bibinfo {booktitle}
  {Density Functional Theory: Modeling, Mathematical Analysis, Computational
  Methods, and Applications}}},\ \bibinfo {editor} {edited by\ \bibinfo
  {editor} {\bibfnamefont {E.}~\bibnamefont {Canc{\`e}s}}\ and\ \bibinfo
  {editor} {\bibfnamefont {G.}~\bibnamefont {Friesecke}}}\ (\bibinfo
  {publisher} {Springer},\ \bibinfo {year} {2023})\ pp.\ \bibinfo {pages}
  {115--182}\BibitemShut {NoStop}%
\bibitem [{\citenamefont {Teale}\ \emph {et~al.}(2022)\citenamefont {Teale},
  \citenamefont {Helgaker}, \citenamefont {Savin}, \citenamefont {Adamo},
  \citenamefont {Aradi}, \citenamefont {Arbuznikov}, \citenamefont {Ayers},
  \citenamefont {Baerends}, \citenamefont {Barone}, \citenamefont {Calaminici}
  \emph {et~al.}}]{teale2022dft}%
  \BibitemOpen
  \bibfield  {author} {\bibinfo {author} {\bibfnamefont {A.~M.}\ \bibnamefont
  {Teale}}, \bibinfo {author} {\bibfnamefont {T.}~\bibnamefont {Helgaker}},
  \bibinfo {author} {\bibfnamefont {A.}~\bibnamefont {Savin}}, \bibinfo
  {author} {\bibfnamefont {C.}~\bibnamefont {Adamo}}, \bibinfo {author}
  {\bibfnamefont {B.}~\bibnamefont {Aradi}}, \bibinfo {author} {\bibfnamefont
  {A.~V.}\ \bibnamefont {Arbuznikov}}, \bibinfo {author} {\bibfnamefont
  {P.~W.}\ \bibnamefont {Ayers}}, \bibinfo {author} {\bibfnamefont {E.~J.}\
  \bibnamefont {Baerends}}, \bibinfo {author} {\bibfnamefont {V.}~\bibnamefont
  {Barone}}, \bibinfo {author} {\bibfnamefont {P.}~\bibnamefont {Calaminici}},
  \emph {et~al.},\ }\bibfield  {title} {\emph {\bibinfo {title} {{DFT}
  exchange: sharing perspectives on the workhorse of quantum chemistry and
  materials science},\ }}\href {\doibase 10.1039/D2CP02827A} {\bibfield
  {journal} {\bibinfo  {journal} {Phys. Chem. Chem. Phys.}\ }\textbf {\bibinfo
  {volume} {24}},\ \bibinfo {pages} {28700} (\bibinfo {year}
  {2022})}\BibitemShut {NoStop}%
\bibitem [{\citenamefont {Chayes}\ \emph {et~al.}(1985)\citenamefont {Chayes},
  \citenamefont {Chayes},\ and\ \citenamefont {Ruskai}}]{CCR1985}%
  \BibitemOpen
  \bibfield  {author} {\bibinfo {author} {\bibfnamefont {J.~T.}\ \bibnamefont
  {Chayes}}, \bibinfo {author} {\bibfnamefont {L.}~\bibnamefont {Chayes}}, \
  and\ \bibinfo {author} {\bibfnamefont {M.~B.}\ \bibnamefont {Ruskai}},\
  }\bibfield  {title} {\emph {\bibinfo {title} {Density functional approach to
  quantum lattice systems},\ }}\href {\doibase 10.1007/BF01010474} {\bibfield
  {journal} {\bibinfo  {journal} {J. Stat. Phys.}\ }\textbf {\bibinfo {volume}
  {38}},\ \bibinfo {pages} {497} (\bibinfo {year} {1985})}\BibitemShut
  {NoStop}%
\bibitem [{\citenamefont {Penz}\ and\ \citenamefont {van
  Leeuwen}(2021)}]{penz-DFT-graphs}%
  \BibitemOpen
  \bibfield  {author} {\bibinfo {author} {\bibfnamefont {M.}~\bibnamefont
  {Penz}}\ and\ \bibinfo {author} {\bibfnamefont {R.}~\bibnamefont {van
  Leeuwen}},\ }\bibfield  {title} {\emph {\bibinfo {title} {Density-functional
  theory on graphs},\ }}\href {\doibase 10.1063/5.0074249} {\bibfield
  {journal} {\bibinfo  {journal} {J. Chem. Phys.}\ }\textbf {\bibinfo {volume}
  {155}},\ \bibinfo {pages} {244111} (\bibinfo {year} {2021})}\BibitemShut
  {NoStop}%
\bibitem [{\citenamefont {Lammert}(2010)}]{Lammert2010}%
  \BibitemOpen
  \bibfield  {author} {\bibinfo {author} {\bibfnamefont {P.~E.}\ \bibnamefont
  {Lammert}},\ }\bibfield  {title} {\emph {\bibinfo {title} {Well-behaved
  coarse-grained model of density-functional theory},\ }}\href {\doibase
  10.1103/PhysRevA.82.012109} {\bibfield  {journal} {\bibinfo  {journal} {Phys.
  Rev. A}\ }\textbf {\bibinfo {volume} {82}},\ \bibinfo {pages} {012109}
  (\bibinfo {year} {2010})}\BibitemShut {NoStop}%
\bibitem [{\citenamefont {Lammert}(2007)}]{Lammert2007}%
  \BibitemOpen
  \bibfield  {author} {\bibinfo {author} {\bibfnamefont {P.~E.}\ \bibnamefont
  {Lammert}},\ }\bibfield  {title} {\emph {\bibinfo {title} {Differentiability
  of {L}ieb functional in electronic density functional theory},\ }}\href
  {\doibase 10.1002/qua.21342} {\bibfield  {journal} {\bibinfo  {journal} {Int.
  J. Quantum Chem.}\ }\textbf {\bibinfo {volume} {107}},\ \bibinfo {pages}
  {1943} (\bibinfo {year} {2007})}\BibitemShut {NoStop}%
\bibitem [{\citenamefont {Aryasetiawan}\ and\ \citenamefont
  {Stott}(1986)}]{aryasetiawan-stott1986}%
  \BibitemOpen
  \bibfield  {author} {\bibinfo {author} {\bibfnamefont {F.}~\bibnamefont
  {Aryasetiawan}}\ and\ \bibinfo {author} {\bibfnamefont {M.}~\bibnamefont
  {Stott}},\ }\bibfield  {title} {\emph {\bibinfo {title} {Density-functional
  theory for two noninteracting spinless fermions},\ }}\href {\doibase
  10.1103/PhysRevB.34.4401} {\bibfield  {journal} {\bibinfo  {journal} {Phys.
  Rev. B}\ }\textbf {\bibinfo {volume} {34}},\ \bibinfo {pages} {4401}
  (\bibinfo {year} {1986})}\BibitemShut {NoStop}%
\bibitem [{\citenamefont {Aryasetiawan}\ and\ \citenamefont
  {Stott}(1988)}]{aryasetiawan-stott1988}%
  \BibitemOpen
  \bibfield  {author} {\bibinfo {author} {\bibfnamefont {F.}~\bibnamefont
  {Aryasetiawan}}\ and\ \bibinfo {author} {\bibfnamefont {M.}~\bibnamefont
  {Stott}},\ }\bibfield  {title} {\emph {\bibinfo {title} {Effective potentials
  in density-functional theory},\ }}\href {\doibase 10.1103/PhysRevB.38.2974}
  {\bibfield  {journal} {\bibinfo  {journal} {Phys. Rev. B}\ }\textbf {\bibinfo
  {volume} {38}},\ \bibinfo {pages} {2974} (\bibinfo {year}
  {1988})}\BibitemShut {NoStop}%
\bibitem [{\citenamefont
  {Aryasetiawan}(1989)}]{aryasetiawan-1989-personal-comm}%
  \BibitemOpen
  \bibfield  {author} {\bibinfo {author} {\bibfnamefont {F.}~\bibnamefont
  {Aryasetiawan}},\ }\href@noop {} {\bibfield  {title} {\emph {\bibinfo {title}
  {The v-representability problem in one dimension},\ }}}\bibinfo
  {howpublished} {personal communication} (\bibinfo {year} {1989})\BibitemShut
  {NoStop}%
\bibitem [{\citenamefont {Chen}\ and\ \citenamefont
  {Stott}(1991{\natexlab{a}})}]{chen-stott1991-few-fermions}%
  \BibitemOpen
  \bibfield  {author} {\bibinfo {author} {\bibfnamefont {J.}~\bibnamefont
  {Chen}}\ and\ \bibinfo {author} {\bibfnamefont {M.}~\bibnamefont {Stott}},\
  }\bibfield  {title} {\emph {\bibinfo {title} {v-representability for systems
  of a few fermions},\ }}\href {\doibase 10.1103/PhysRevA.44.2809} {\bibfield
  {journal} {\bibinfo  {journal} {Phys. Rev. A}\ }\textbf {\bibinfo {volume}
  {44}},\ \bibinfo {pages} {2809} (\bibinfo {year}
  {1991}{\natexlab{a}})}\BibitemShut {NoStop}%
\bibitem [{\citenamefont {Chen}\ and\ \citenamefont
  {Stott}(1991{\natexlab{b}})}]{chen-stott1991-low-deg}%
  \BibitemOpen
  \bibfield  {author} {\bibinfo {author} {\bibfnamefont {J.}~\bibnamefont
  {Chen}}\ and\ \bibinfo {author} {\bibfnamefont {M.}~\bibnamefont {Stott}},\
  }\bibfield  {title} {\emph {\bibinfo {title} {v-representability for systems
  with low degeneracy},\ }}\href {\doibase 10.1103/PhysRevA.44.2816} {\bibfield
   {journal} {\bibinfo  {journal} {Phys. Rev. A}\ }\textbf {\bibinfo {volume}
  {44}},\ \bibinfo {pages} {2816} (\bibinfo {year}
  {1991}{\natexlab{b}})}\BibitemShut {NoStop}%
\bibitem [{\citenamefont {Chen}\ and\ \citenamefont
  {Stott}(1993)}]{chen-stott1993-v-rep}%
  \BibitemOpen
  \bibfield  {author} {\bibinfo {author} {\bibfnamefont {J.}~\bibnamefont
  {Chen}}\ and\ \bibinfo {author} {\bibfnamefont {M.}~\bibnamefont {Stott}},\
  }\bibfield  {title} {\emph {\bibinfo {title} {v-representability for
  noninteracting systems},\ }}\href {\doibase 10.1103/PhysRevA.47.153}
  {\bibfield  {journal} {\bibinfo  {journal} {Phys. Rev. A}\ }\textbf {\bibinfo
  {volume} {47}},\ \bibinfo {pages} {153} (\bibinfo {year} {1993})}\BibitemShut
  {NoStop}%
\bibitem [{\citenamefont {Garrigue}(2022)}]{garrigue2022building}%
  \BibitemOpen
  \bibfield  {author} {\bibinfo {author} {\bibfnamefont {L.}~\bibnamefont
  {Garrigue}},\ }\bibfield  {title} {\emph {\bibinfo {title} {Building
  {K}ohn--{S}ham potentials for ground and excited states},\ }}\href {\doibase
  10.1007/s00205-022-01804-1} {\bibfield  {journal} {\bibinfo  {journal} {Arch.
  Ration. Mech. Anal.}\ }\textbf {\bibinfo {volume} {245}},\ \bibinfo {pages}
  {949} (\bibinfo {year} {2022})}\BibitemShut {NoStop}%
\bibitem [{\citenamefont {Lieb}\ and\ \citenamefont
  {Loss}(2001)}]{Lieb-Loss-book}%
  \BibitemOpen
  \bibfield  {author} {\bibinfo {author} {\bibfnamefont {E.~H.}\ \bibnamefont
  {Lieb}}\ and\ \bibinfo {author} {\bibfnamefont {M.}~\bibnamefont {Loss}},\
  }\href@noop {} {\emph {\bibinfo {title} {Analysis}}},\ \bibinfo {edition}
  {2nd}\ ed.\ (\bibinfo  {publisher} {American Mathematical Society},\ \bibinfo
  {year} {2001})\BibitemShut {NoStop}%
\bibitem [{\citenamefont {Reed}\ and\ \citenamefont
  {Simon}(1980)}]{reed-simon-1}%
  \BibitemOpen
  \bibfield  {author} {\bibinfo {author} {\bibfnamefont {M.}~\bibnamefont
  {Reed}}\ and\ \bibinfo {author} {\bibfnamefont {B.}~\bibnamefont {Simon}},\
  }\href@noop {} {\emph {\bibinfo {title} {Methods of Modern Mathematical
  Physics. Vol. I: Functional analysis}}},\ \bibinfo {edition} {2nd}\ ed.,\
  Vol.~\bibinfo {volume} {1}\ (\bibinfo  {publisher} {Academic Press},\
  \bibinfo {year} {1980})\BibitemShut {NoStop}%
\bibitem [{\citenamefont {Reed}\ and\ \citenamefont
  {Simon}(1975)}]{reed-simon-2}%
  \BibitemOpen
  \bibfield  {author} {\bibinfo {author} {\bibfnamefont {M.}~\bibnamefont
  {Reed}}\ and\ \bibinfo {author} {\bibfnamefont {B.}~\bibnamefont {Simon}},\
  }\href@noop {} {\emph {\bibinfo {title} {Methods of Modern Mathematical
  Physics. Vol. II: Fourier Analysis, Self-Adjointness}}},\ Vol.~\bibinfo
  {volume} {2}\ (\bibinfo  {publisher} {Academic Press},\ \bibinfo {year}
  {1975})\BibitemShut {NoStop}%
\bibitem [{\citenamefont {Brezis}(2011)}]{Brezis-book}%
  \BibitemOpen
  \bibfield  {author} {\bibinfo {author} {\bibfnamefont {H.}~\bibnamefont
  {Brezis}},\ }\href@noop {} {\emph {\bibinfo {title} {Functional Analysis,
  Sobolev Spaces and Partial Differential Equations}}}\ (\bibinfo  {publisher}
  {Springer},\ \bibinfo {year} {2011})\BibitemShut {NoStop}%
\bibitem [{\citenamefont {Adams}\ and\ \citenamefont
  {Fournier}(2003)}]{adams-book}%
  \BibitemOpen
  \bibfield  {author} {\bibinfo {author} {\bibfnamefont {R.~A.}\ \bibnamefont
  {Adams}}\ and\ \bibinfo {author} {\bibfnamefont {J.~J.}\ \bibnamefont
  {Fournier}},\ }\href@noop {} {\emph {\bibinfo {title} {Sobolev spaces}}},\
  \bibinfo {edition} {2nd}\ ed.\ (\bibinfo  {publisher} {Elsevier},\ \bibinfo
  {year} {2003})\BibitemShut {NoStop}%
\bibitem [{\citenamefont {Reed}\ and\ \citenamefont
  {Simon}(1978)}]{reed-simon-4}%
  \BibitemOpen
  \bibfield  {author} {\bibinfo {author} {\bibfnamefont {M.}~\bibnamefont
  {Reed}}\ and\ \bibinfo {author} {\bibfnamefont {B.}~\bibnamefont {Simon}},\
  }\href@noop {} {\emph {\bibinfo {title} {Methods of Modern Mathematical
  Physics. Vol. IV: Analysis of Operators}}},\ Vol.~\bibinfo {volume} {4}\
  (\bibinfo  {publisher} {Academic Press},\ \bibinfo {year} {1978})\BibitemShut
  {NoStop}%
\bibitem [{\citenamefont {Barbu}\ and\ \citenamefont
  {Precupanu}(2012)}]{Barbu-Precupanu}%
  \BibitemOpen
  \bibfield  {author} {\bibinfo {author} {\bibfnamefont {V.}~\bibnamefont
  {Barbu}}\ and\ \bibinfo {author} {\bibfnamefont {T.}~\bibnamefont
  {Precupanu}},\ }\href@noop {} {\emph {\bibinfo {title} {Convexity and
  Optimization in Banach Spaces}}},\ \bibinfo {edition} {4th}\ ed.\ (\bibinfo
  {publisher} {Springer},\ \bibinfo {year} {2012})\BibitemShut {NoStop}%
\bibitem [{\citenamefont {Valone}(1980)}]{valone1980b}%
  \BibitemOpen
  \bibfield  {author} {\bibinfo {author} {\bibfnamefont {S.~M.}\ \bibnamefont
  {Valone}},\ }\bibfield  {title} {\emph {\bibinfo {title} {A one-to-one
  mapping between one-particle densities and some n-particle ensembles},\
  }}\href {\doibase 10.1063/1.440656} {\bibfield  {journal} {\bibinfo
  {journal} {J. Chem. Phys.}\ }\textbf {\bibinfo {volume} {73}},\ \bibinfo
  {pages} {4653} (\bibinfo {year} {1980})}\BibitemShut {NoStop}%
\bibitem [{\citenamefont {Herczy\'nski}(1989)}]{HERCZYNSKI1989-KLMN}%
  \BibitemOpen
  \bibfield  {author} {\bibinfo {author} {\bibfnamefont {J.}~\bibnamefont
  {Herczy\'nski}},\ }\bibfield  {title} {\emph {\bibinfo {title} {On
  {S}chr\"odinger operators with distributional potentials},\ }}\href
  {http://www.jstor.org/stable/24714447} {\bibfield  {journal} {\bibinfo
  {journal} {J. Operator Theory}\ }\textbf {\bibinfo {volume} {21}},\ \bibinfo
  {pages} {273} (\bibinfo {year} {1989})}\BibitemShut {NoStop}%
\bibitem [{\citenamefont {Clason}(2020)}]{book-clason}%
  \BibitemOpen
  \bibfield  {author} {\bibinfo {author} {\bibfnamefont {C.}~\bibnamefont
  {Clason}},\ }\href@noop {} {\emph {\bibinfo {title} {Introduction to
  Functional Analysis}}}\ (\bibinfo  {publisher} {Springer Nature},\ \bibinfo
  {year} {2020})\BibitemShut {NoStop}%
\bibitem [{\citenamefont {Kronig}\ and\ \citenamefont {Penney}(1931)}]{KP1931}%
  \BibitemOpen
  \bibfield  {author} {\bibinfo {author} {\bibfnamefont {R.~D.~L.}\
  \bibnamefont {Kronig}}\ and\ \bibinfo {author} {\bibfnamefont {W.~G.}\
  \bibnamefont {Penney}},\ }\bibfield  {title} {\emph {\bibinfo {title}
  {Quantum mechanics of electrons in crystal lattices},\ }}\href {\doibase
  10.1098/rspa.1931.0019} {\bibfield  {journal} {\bibinfo  {journal} {Proc. R.
  Soc. Lond. Ser. A}\ }\textbf {\bibinfo {volume} {130}},\ \bibinfo {pages}
  {499–513} (\bibinfo {year} {1931})}\BibitemShut {NoStop}%
\bibitem [{\citenamefont {Kostenko}\ and\ \citenamefont
  {Malamud}(2010)}]{Kostenko2010}%
  \BibitemOpen
  \bibfield  {author} {\bibinfo {author} {\bibfnamefont {A.~S.}\ \bibnamefont
  {Kostenko}}\ and\ \bibinfo {author} {\bibfnamefont {M.~M.}\ \bibnamefont
  {Malamud}},\ }\bibfield  {title} {\emph {\bibinfo {title} {{1-D}
  {S}chr\"{o}dinger operators with local point interactions on a discrete
  set},\ }}\href {\doibase 10.1016/j.jde.2010.02.011} {\bibfield  {journal}
  {\bibinfo  {journal} {J. Differ. Equ.}\ }\textbf {\bibinfo {volume} {249}},\
  \bibinfo {pages} {253–304} (\bibinfo {year} {2010})}\BibitemShut {NoStop}%
\bibitem [{\citenamefont {Albeverio}\ \emph {et~al.}(2012)\citenamefont
  {Albeverio}, \citenamefont {Gesztesy}, \citenamefont {Hoegh-Krohn},\ and\
  \citenamefont {Holden}}]{albeverio2012solvable}%
  \BibitemOpen
  \bibfield  {author} {\bibinfo {author} {\bibfnamefont {S.}~\bibnamefont
  {Albeverio}}, \bibinfo {author} {\bibfnamefont {F.}~\bibnamefont {Gesztesy}},
  \bibinfo {author} {\bibfnamefont {R.}~\bibnamefont {Hoegh-Krohn}}, \ and\
  \bibinfo {author} {\bibfnamefont {H.}~\bibnamefont {Holden}},\ }\href@noop {}
  {\emph {\bibinfo {title} {Solvable models in quantum mechanics}}}\ (\bibinfo
  {publisher} {Springer},\ \bibinfo {year} {2012})\BibitemShut {NoStop}%
\bibitem [{\citenamefont {Haldane}(1983)}]{Haldane1983}%
  \BibitemOpen
  \bibfield  {author} {\bibinfo {author} {\bibfnamefont {F.~D.~M.}\
  \bibnamefont {Haldane}},\ }\bibfield  {title} {\emph {\bibinfo {title}
  {Fractional quantization of the {H}all effect: A hierarchy of incompressible
  quantum fluid states},\ }}\href {\doibase 10.1103/physrevlett.51.605}
  {\bibfield  {journal} {\bibinfo  {journal} {Phys. Rev. Lett.}\ }\textbf
  {\bibinfo {volume} {51}},\ \bibinfo {pages} {605–608} (\bibinfo {year}
  {1983})}\BibitemShut {NoStop}%
\bibitem [{\citenamefont {Seiringer}\ and\ \citenamefont
  {Yngvason}(2020)}]{Seiringer2020}%
  \BibitemOpen
  \bibfield  {author} {\bibinfo {author} {\bibfnamefont {R.}~\bibnamefont
  {Seiringer}}\ and\ \bibinfo {author} {\bibfnamefont {J.}~\bibnamefont
  {Yngvason}},\ }\bibfield  {title} {\emph {\bibinfo {title} {Emergence of
  {H}aldane pseudo-potentials in systems with short-range interactions},\
  }}\href {\doibase 10.1007/s10955-020-02586-0} {\bibfield  {journal} {\bibinfo
   {journal} {J. Stat. Phys.}\ }\textbf {\bibinfo {volume} {181}},\ \bibinfo
  {pages} {448–464} (\bibinfo {year} {2020})}\BibitemShut {NoStop}%
\bibitem [{\citenamefont {Lieb}\ and\ \citenamefont
  {Liniger}(1963)}]{Lieb1963}%
  \BibitemOpen
  \bibfield  {author} {\bibinfo {author} {\bibfnamefont {E.~H.}\ \bibnamefont
  {Lieb}}\ and\ \bibinfo {author} {\bibfnamefont {W.}~\bibnamefont {Liniger}},\
  }\bibfield  {title} {\emph {\bibinfo {title} {Exact analysis of an
  interacting {B}ose gas. {I.} {T}he general solution and the ground state},\
  }}\href {\doibase 10.1103/physrev.130.1605} {\bibfield  {journal} {\bibinfo
  {journal} {Phys. Rev.}\ }\textbf {\bibinfo {volume} {130}},\ \bibinfo {pages}
  {1605–1616} (\bibinfo {year} {1963})}\BibitemShut {NoStop}%
\bibitem [{\citenamefont {Kopyciński}\ \emph {et~al.}(2022)\citenamefont
  {Kopyciński}, \citenamefont {Łebek}, \citenamefont {Marciniak},
  \citenamefont {Ołdziejewski}, \citenamefont {Górecki},\ and\ \citenamefont
  {Pawłowski}}]{Kopyciski2022}%
  \BibitemOpen
  \bibfield  {author} {\bibinfo {author} {\bibfnamefont {J.}~\bibnamefont
  {Kopyciński}}, \bibinfo {author} {\bibfnamefont {M.}~\bibnamefont {Łebek}},
  \bibinfo {author} {\bibfnamefont {M.}~\bibnamefont {Marciniak}}, \bibinfo
  {author} {\bibfnamefont {R.}~\bibnamefont {Ołdziejewski}}, \bibinfo {author}
  {\bibfnamefont {W.}~\bibnamefont {Górecki}}, \ and\ \bibinfo {author}
  {\bibfnamefont {K.}~\bibnamefont {Pawłowski}},\ }\bibfield  {title} {\emph
  {\bibinfo {title} {Beyond {G}ross--{P}itaevskii equation for {1D} gas:
  quasiparticles and solitons},\ }}\href {\doibase
  10.21468/scipostphys.12.1.023} {\bibfield  {journal} {\bibinfo  {journal}
  {SciPost Physics}\ }\textbf {\bibinfo {volume} {12}} (\bibinfo {year}
  {2022}),\ 10.21468/scipostphys.12.1.023}\BibitemShut {NoStop}%
\bibitem [{\citenamefont {Penz}\ \emph
  {et~al.}(2023{\natexlab{a}})\citenamefont {Penz}, \citenamefont {Tellgren},
  \citenamefont {Csirik}, \citenamefont {Ruggenthaler},\ and\ \citenamefont
  {Laestadius}}]{Penz-et-al-HKreview-partI}%
  \BibitemOpen
  \bibfield  {author} {\bibinfo {author} {\bibfnamefont {M.}~\bibnamefont
  {Penz}}, \bibinfo {author} {\bibfnamefont {E.~I.}\ \bibnamefont {Tellgren}},
  \bibinfo {author} {\bibfnamefont {M.~A.}\ \bibnamefont {Csirik}}, \bibinfo
  {author} {\bibfnamefont {M.}~\bibnamefont {Ruggenthaler}}, \ and\ \bibinfo
  {author} {\bibfnamefont {A.}~\bibnamefont {Laestadius}},\ }\bibfield  {title}
  {\emph {\bibinfo {title} {The structure of density-potential mapping. {P}art
  {I}: Standard density-functional theory},\ }}\href {\doibase
  10.1021/acsphyschemau.2c00069} {\bibfield  {journal} {\bibinfo  {journal}
  {ACS Phys. Chem. Au}\ }\textbf {\bibinfo {volume} {3}},\ \bibinfo {pages}
  {334} (\bibinfo {year} {2023}{\natexlab{a}})}\BibitemShut {NoStop}%
\bibitem [{\citenamefont {Garrigue}(2018)}]{Garrigue2018HK}%
  \BibitemOpen
  \bibfield  {author} {\bibinfo {author} {\bibfnamefont {L.}~\bibnamefont
  {Garrigue}},\ }\bibfield  {title} {\emph {\bibinfo {title} {Unique
  continuation for many-body {S}chrödinger operators and the
  {H}ohenberg-{K}ohn theorem},\ }}\href {\doibase 10.1007/s11040-018-9287-z}
  {\bibfield  {journal} {\bibinfo  {journal} {Math. Phys. Anal. Geom.}\
  }\textbf {\bibinfo {volume} {21}},\ \bibinfo {pages} {27} (\bibinfo {year}
  {2018})}\BibitemShut {NoStop}%
\bibitem [{\citenamefont {Dreizler}\ and\ \citenamefont
  {Gross}(2012)}]{Dreizler-Gross}%
  \BibitemOpen
  \bibfield  {author} {\bibinfo {author} {\bibfnamefont {R.~M.}\ \bibnamefont
  {Dreizler}}\ and\ \bibinfo {author} {\bibfnamefont {E.~K.}\ \bibnamefont
  {Gross}},\ }\href@noop {} {\emph {\bibinfo {title} {Density functional
  theory: an approach to the quantum many-body problem}}}\ (\bibinfo
  {publisher} {Springer},\ \bibinfo {year} {2012})\BibitemShut {NoStop}%
\bibitem [{\citenamefont {Langreth}\ and\ \citenamefont
  {Mehl}(1983)}]{LangrethMehl1983}%
  \BibitemOpen
  \bibfield  {author} {\bibinfo {author} {\bibfnamefont {D.~C.}\ \bibnamefont
  {Langreth}}\ and\ \bibinfo {author} {\bibfnamefont {M.~J.}\ \bibnamefont
  {Mehl}},\ }\bibfield  {title} {\emph {\bibinfo {title} {Beyond the
  local-density approximation in calculations of ground-state electronic
  properties},\ }}\href {\doibase 10.1103/PhysRevB.28.1809} {\bibfield
  {journal} {\bibinfo  {journal} {Phys. Rev. B}\ }\textbf {\bibinfo {volume}
  {28}},\ \bibinfo {pages} {1809} (\bibinfo {year} {1983})}\BibitemShut
  {NoStop}%
\bibitem [{\citenamefont {Svendsen}\ and\ \citenamefont {von
  Barth}(1996)}]{SvendsenBarth1996}%
  \BibitemOpen
  \bibfield  {author} {\bibinfo {author} {\bibfnamefont {P.~S.}\ \bibnamefont
  {Svendsen}}\ and\ \bibinfo {author} {\bibfnamefont {U.}~\bibnamefont {von
  Barth}},\ }\bibfield  {title} {\emph {\bibinfo {title} {Gradient expansion of
  the exchange energy from second-order density response theory},\ }}\href
  {\doibase 10.1103/PhysRevB.54.17402} {\bibfield  {journal} {\bibinfo
  {journal} {Phys. Rev. B}\ }\textbf {\bibinfo {volume} {54}},\ \bibinfo
  {pages} {17402} (\bibinfo {year} {1996})}\BibitemShut {NoStop}%
\bibitem [{\citenamefont {Perdew}(1985)}]{Perdew1985GGA}%
  \BibitemOpen
  \bibfield  {author} {\bibinfo {author} {\bibfnamefont {J.~P.}\ \bibnamefont
  {Perdew}},\ }\bibfield  {title} {\emph {\bibinfo {title} {Accurate density
  functional for the energy: Real-space cutoff of the gradient expansion for
  the exchange hole},\ }}\href {\doibase 10.1103/PhysRevLett.55.1665}
  {\bibfield  {journal} {\bibinfo  {journal} {Phys. Rev. Lett.}\ }\textbf
  {\bibinfo {volume} {55}},\ \bibinfo {pages} {1665} (\bibinfo {year}
  {1985})}\BibitemShut {NoStop}%
\bibitem [{\citenamefont {Perdew}(1986)}]{Perdew1986GGA}%
  \BibitemOpen
  \bibfield  {author} {\bibinfo {author} {\bibfnamefont {J.~P.}\ \bibnamefont
  {Perdew}},\ }\bibfield  {title} {\emph {\bibinfo {title} {Density-functional
  approximation for the correlation energy of the inhomogeneous electron gas},\
  }}\href {\doibase 10.1103/PhysRevB.33.8822} {\bibfield  {journal} {\bibinfo
  {journal} {Phys. Rev. B}\ }\textbf {\bibinfo {volume} {33}},\ \bibinfo
  {pages} {8822} (\bibinfo {year} {1986})}\BibitemShut {NoStop}%
\bibitem [{\citenamefont {Becke}(1988)}]{Becke1988GGA}%
  \BibitemOpen
  \bibfield  {author} {\bibinfo {author} {\bibfnamefont {A.~D.}\ \bibnamefont
  {Becke}},\ }\bibfield  {title} {\emph {\bibinfo {title} {Density-functional
  exchange-energy approximation with correct asymptotic behavior},\ }}\href
  {\doibase 10.1103/PhysRevA.38.3098} {\bibfield  {journal} {\bibinfo
  {journal} {Phys. Rev. A}\ }\textbf {\bibinfo {volume} {38}},\ \bibinfo
  {pages} {3098} (\bibinfo {year} {1988})}\BibitemShut {NoStop}%
\bibitem [{\citenamefont {Csonka}\ \emph {et~al.}(2010)\citenamefont {Csonka},
  \citenamefont {Perdew},\ and\ \citenamefont {Ruzsinszky}}]{Csonka2010}%
  \BibitemOpen
  \bibfield  {author} {\bibinfo {author} {\bibfnamefont {G.~I.}\ \bibnamefont
  {Csonka}}, \bibinfo {author} {\bibfnamefont {J.~P.}\ \bibnamefont {Perdew}},
  \ and\ \bibinfo {author} {\bibfnamefont {A.}~\bibnamefont {Ruzsinszky}},\
  }\bibfield  {title} {\emph {\bibinfo {title} {Global hybrid functionals: A
  look at the engine under the hood},\ }}\href {\doibase 10.1021/ct100488v}
  {\bibfield  {journal} {\bibinfo  {journal} {J. Chem. Theory Comput.}\
  }\textbf {\bibinfo {volume} {6}},\ \bibinfo {pages} {3688–3703} (\bibinfo
  {year} {2010})}\BibitemShut {NoStop}%
\bibitem [{\citenamefont {Seidl}\ \emph {et~al.}(1996)\citenamefont {Seidl},
  \citenamefont {G\"{o}rling}, \citenamefont {Vogl}, \citenamefont {Majewski},\
  and\ \citenamefont {Levy}}]{Seidl1996}%
  \BibitemOpen
  \bibfield  {author} {\bibinfo {author} {\bibfnamefont {A.}~\bibnamefont
  {Seidl}}, \bibinfo {author} {\bibfnamefont {A.}~\bibnamefont {G\"{o}rling}},
  \bibinfo {author} {\bibfnamefont {P.}~\bibnamefont {Vogl}}, \bibinfo {author}
  {\bibfnamefont {J.~A.}\ \bibnamefont {Majewski}}, \ and\ \bibinfo {author}
  {\bibfnamefont {M.}~\bibnamefont {Levy}},\ }\bibfield  {title} {\emph
  {\bibinfo {title} {Generalized {K}ohn--{S}ham schemes and the band-gap
  problem},\ }}\href {\doibase 10.1103/physrevb.53.3764} {\bibfield  {journal}
  {\bibinfo  {journal} {Phys. Rev. B}\ }\textbf {\bibinfo {volume} {53}},\
  \bibinfo {pages} {3764–3774} (\bibinfo {year} {1996})}\BibitemShut
  {NoStop}%
\bibitem [{\citenamefont {Garrick}\ \emph {et~al.}(2020)\citenamefont
  {Garrick}, \citenamefont {Natan}, \citenamefont {Gould},\ and\ \citenamefont
  {Kronik}}]{Garrick2020}%
  \BibitemOpen
  \bibfield  {author} {\bibinfo {author} {\bibfnamefont {R.}~\bibnamefont
  {Garrick}}, \bibinfo {author} {\bibfnamefont {A.}~\bibnamefont {Natan}},
  \bibinfo {author} {\bibfnamefont {T.}~\bibnamefont {Gould}}, \ and\ \bibinfo
  {author} {\bibfnamefont {L.}~\bibnamefont {Kronik}},\ }\bibfield  {title}
  {\emph {\bibinfo {title} {Exact generalized {K}ohn--{S}ham theory for hybrid
  functionals},\ }}\href {\doibase 10.1103/physrevx.10.021040} {\bibfield
  {journal} {\bibinfo  {journal} {Phys. Rev. X}\ }\textbf {\bibinfo {volume}
  {10}} (\bibinfo {year} {2020}),\ 10.1103/physrevx.10.021040}\BibitemShut
  {NoStop}%
\bibitem [{\citenamefont {Penz}\ \emph {et~al.}(2019)\citenamefont {Penz},
  \citenamefont {Laestadius}, \citenamefont {Tellgren},\ and\ \citenamefont
  {Ruggenthaler}}]{penz2019guaranteed}%
  \BibitemOpen
  \bibfield  {author} {\bibinfo {author} {\bibfnamefont {M.}~\bibnamefont
  {Penz}}, \bibinfo {author} {\bibfnamefont {A.}~\bibnamefont {Laestadius}},
  \bibinfo {author} {\bibfnamefont {E.~I.}\ \bibnamefont {Tellgren}}, \ and\
  \bibinfo {author} {\bibfnamefont {M.}~\bibnamefont {Ruggenthaler}},\
  }\bibfield  {title} {\emph {\bibinfo {title} {Guaranteed convergence of a
  regularized {K}ohn--{S}ham iteration in finite dimensions},\ }}\href
  {\doibase 10.1103/physrevlett.123.037401} {\bibfield  {journal} {\bibinfo
  {journal} {Phys. Rev. Lett.}\ }\textbf {\bibinfo {volume} {123}},\ \bibinfo
  {pages} {037401} (\bibinfo {year} {2019})}\BibitemShut {NoStop}%
\bibitem [{\citenamefont {Penz}\ \emph {et~al.}(2020)\citenamefont {Penz},
  \citenamefont {Laestadius}, \citenamefont {Tellgren}, \citenamefont
  {Ruggenthaler},\ and\ \citenamefont {Lammert}}]{penz2020erratum}%
  \BibitemOpen
  \bibfield  {author} {\bibinfo {author} {\bibfnamefont {M.}~\bibnamefont
  {Penz}}, \bibinfo {author} {\bibfnamefont {A.}~\bibnamefont {Laestadius}},
  \bibinfo {author} {\bibfnamefont {E.~I.}\ \bibnamefont {Tellgren}}, \bibinfo
  {author} {\bibfnamefont {M.}~\bibnamefont {Ruggenthaler}}, \ and\ \bibinfo
  {author} {\bibfnamefont {P.~E.}\ \bibnamefont {Lammert}},\ }\bibfield
  {title} {\emph {\bibinfo {title} {Erratum: {G}uaranteed convergence of a
  regularized {K}ohn--{S}ham iteration in finite dimensions},\ }}\href
  {\doibase 10.1103/PhysRevLett.125.249902} {\bibfield  {journal} {\bibinfo
  {journal} {Phys. Rev. Lett.}\ }\textbf {\bibinfo {volume} {125}},\ \bibinfo
  {pages} {249902} (\bibinfo {year} {2020})}\BibitemShut {NoStop}%
\bibitem [{\citenamefont {Lammert}(2023)}]{Lammert-bivariate}%
  \BibitemOpen
  \bibfield  {author} {\bibinfo {author} {\bibfnamefont {P.~E.}\ \bibnamefont
  {Lammert}},\ }\bibfield  {title} {\emph {\bibinfo {title} {{K}ohn–-{S}ham
  computation and the bivariate view of density functional theory},\ }}\href
  {\doibase 10.1088/1751-8121/ad075d} {\bibfield  {journal} {\bibinfo
  {journal} {J. Phys. A}\ }\textbf {\bibinfo {volume} {56}},\ \bibinfo {pages}
  {495203} (\bibinfo {year} {2023})}\BibitemShut {NoStop}%
\bibitem [{\citenamefont {Penz}\ \emph
  {et~al.}(2023{\natexlab{b}})\citenamefont {Penz}, \citenamefont {Csirik},\
  and\ \citenamefont {Laestadius}}]{penz2023MY-ZMP}%
  \BibitemOpen
  \bibfield  {author} {\bibinfo {author} {\bibfnamefont {M.}~\bibnamefont
  {Penz}}, \bibinfo {author} {\bibfnamefont {M.~A.}\ \bibnamefont {Csirik}}, \
  and\ \bibinfo {author} {\bibfnamefont {A.}~\bibnamefont {Laestadius}},\
  }\bibfield  {title} {\emph {\bibinfo {title} {Density-potential inversion
  from {M}oreau--{Y}osida regularization},\ }}\href {\doibase
  10.1088/2516-1075/acc626} {\bibfield  {journal} {\bibinfo  {journal}
  {Electron. Struct.}\ }\textbf {\bibinfo {volume} {5}},\ \bibinfo {pages}
  {014009} (\bibinfo {year} {2023}{\natexlab{b}})}\BibitemShut {NoStop}%
\bibitem [{\citenamefont {Andrade}\ and\ \citenamefont
  {Aspuru-Guzik}(2011)}]{Andrade2011}%
  \BibitemOpen
  \bibfield  {author} {\bibinfo {author} {\bibfnamefont {X.}~\bibnamefont
  {Andrade}}\ and\ \bibinfo {author} {\bibfnamefont {A.}~\bibnamefont
  {Aspuru-Guzik}},\ }\bibfield  {title} {\emph {\bibinfo {title} {Prediction of
  the derivative discontinuity in density functional theory from an
  electrostatic description of the exchange and correlation potential},\
  }}\href {\doibase 10.1103/physrevlett.107.183002} {\bibfield  {journal}
  {\bibinfo  {journal} {Phys. Rev. Lett.}\ }\textbf {\bibinfo {volume} {107}},\
  \bibinfo {pages} {183002} (\bibinfo {year} {2011})}\BibitemShut {NoStop}%
\bibitem [{\citenamefont {Seidl}\ \emph {et~al.}(2017)\citenamefont {Seidl},
  \citenamefont {Di~Marino}, \citenamefont {Gerolin}, \citenamefont {Nenna},
  \citenamefont {Giesbertz},\ and\ \citenamefont
  {Gori-Giorgi}}]{seidl2017sphericalSCE}%
  \BibitemOpen
  \bibfield  {author} {\bibinfo {author} {\bibfnamefont {M.}~\bibnamefont
  {Seidl}}, \bibinfo {author} {\bibfnamefont {S.}~\bibnamefont {Di~Marino}},
  \bibinfo {author} {\bibfnamefont {A.}~\bibnamefont {Gerolin}}, \bibinfo
  {author} {\bibfnamefont {L.}~\bibnamefont {Nenna}}, \bibinfo {author}
  {\bibfnamefont {K.~J.~H.}\ \bibnamefont {Giesbertz}}, \ and\ \bibinfo
  {author} {\bibfnamefont {P.}~\bibnamefont {Gori-Giorgi}},\ }\href@noop {}
  {\bibfield  {title} {\emph {\bibinfo {title} {The strictly-correlated
  electron functional for spherically symmetric systems revisited},\
  }}}\bibinfo {howpublished} {(16 Feb 2017) arXiv e-prints [cond-mat.str-el]
  1702.05022} (\bibinfo {year} {2017}),\ \bibinfo {note} {accessed
  2023-11-22}\BibitemShut {NoStop}%
\bibitem [{\citenamefont {Giesbertz}\ and\ \citenamefont
  {Ruggenthaler}(2019)}]{GiesbertzRuggenthaler2019}%
  \BibitemOpen
  \bibfield  {author} {\bibinfo {author} {\bibfnamefont {K.~J.~H.}\
  \bibnamefont {Giesbertz}}\ and\ \bibinfo {author} {\bibfnamefont
  {M.}~\bibnamefont {Ruggenthaler}},\ }\bibfield  {title} {\emph {\bibinfo
  {title} {One-body reduced density-matrix functional theory in finite basis
  sets at elevated temperatures},\ }}\href {\doibase
  10.1016/j.physrep.2019.01.010} {\bibfield  {journal} {\bibinfo  {journal}
  {Phys. Rep.}\ }\textbf {\bibinfo {volume} {806}},\ \bibinfo {pages} {1}
  (\bibinfo {year} {2019})}\BibitemShut {NoStop}%
\bibitem [{\citenamefont {Sutter}\ and\ \citenamefont
  {Giesbertz}(2023)}]{SutterGiesbertz2023}%
  \BibitemOpen
  \bibfield  {author} {\bibinfo {author} {\bibfnamefont {S.~M.}\ \bibnamefont
  {Sutter}}\ and\ \bibinfo {author} {\bibfnamefont {K.~J.~H.}\ \bibnamefont
  {Giesbertz}},\ }\bibfield  {title} {\emph {\bibinfo {title} {One-body reduced
  density-matrix functional theory for the canonical ensemble},\ }}\href
  {\doibase 10.1103/PhysRevA.107.022210} {\bibfield  {journal} {\bibinfo
  {journal} {Phys. Rev. A}\ }\textbf {\bibinfo {volume} {107}},\ \bibinfo
  {pages} {022210} (\bibinfo {year} {2023})}\BibitemShut {NoStop}%
\bibitem [{\citenamefont {Gaudin}(1967)}]{gaudin1967systeme}%
  \BibitemOpen
  \bibfield  {author} {\bibinfo {author} {\bibfnamefont {M.}~\bibnamefont
  {Gaudin}},\ }\bibfield  {title} {\emph {\bibinfo {title} {Un systeme a une
  dimension de fermions en interaction},\ }}\href {\doibase
  10.1016/0375-9601(67)90193-4} {\bibfield  {journal} {\bibinfo  {journal}
  {Phys. Lett. A}\ }\textbf {\bibinfo {volume} {24}},\ \bibinfo {pages} {55}
  (\bibinfo {year} {1967})}\BibitemShut {NoStop}%
\end{thebibliography}
\end{document}